\documentclass[journal]{IEEEtran}

\normalsize

\usepackage{amsmath,amssymb,amsfonts,bm}
\usepackage{algorithmic}
\usepackage{graphicx}
\usepackage{float}
\usepackage{mdwmath}

\usepackage{amssymb}
\usepackage{amsmath}
\usepackage{amsthm}
\usepackage{mathrsfs}

\usepackage{textcomp}
\usepackage{subfigure}
\usepackage{float}
\usepackage{subfig}

\usepackage{soul,color}
\usepackage{cancel}
\usepackage{cite}

\usepackage{xcolor,pict2e}

\theoremstyle{plain}
\newtheorem{thm}{Theorem}

\newtheorem{lem}[thm]{Lemma}
\newtheorem{prop}[thm]{Proposition}

\newtheorem{defn}[thm]{Definition}
\newtheorem{rem}[thm]{Remark}

\newcommand{\norm}[1]{\left\lVert#1\right\rVert}

\ifCLASSINFOpdf

\else

\fi

\begin{document}
%
% paper title
% can use linebreaks \\ within to get better formatting as desired
\title{Unsupervised Learning Discriminative MIG Detectors in Nonhomogeneous Clutter}

\author{Xiaoqiang~Hua,~
        Yusuke~Ono,~
        Linyu~Peng,~\IEEEmembership{Member,~IEEE,}
        ~and~Yuting~ Xu% <-this % stops a space
\thanks{This work was supported by NSFC  (Grant No. 61901479), JSPS KAKENHI (Grant
No. JP20K14365), JST CREST (Grant No. JPMJCR1914), and Keio
Gijuku Fukuzawa Memorial Fund.  We thank the editor and the anonymous referees for their constructive comments. {\it (Corresponding author: Linyu Peng.)}

X. Hua is with the College of Meteorology and Oceanography, and the College of Computer Science and Technology, National University of Defense Technology, Changsha 410073, China (e-mail: hxq712@yeah.net).

Y. Ono and L. Peng are with the Department of Mechanical Engineering, Keio University, Yokohama 223-8522, Japan (e-mail: yuu555yuu@keio.jp; l.peng@mech.keio.ac.jp).

Y. Xu is with the College of Physics, Jilin University, Changchun 130012, China (e-mail: xuyt20@mails.jlu.edu.cn).
}}
%\date{January 19, 2022}
% The paper headers
%\markboth{IEEE Transactions on Communications}%
%{Submitted paper}

% make the title area
\maketitle

\begin{abstract}
%\boldmath
Principal component analysis (PCA) is a commonly used pattern analysis method that maps high-dimensional data into a lower-dimensional space maximizing the data variance, that results in the promotion of separability of data. Inspired by the principle of PCA, a novel type of learning discriminative matrix information geometry (MIG) detectors in the unsupervised scenario are developed, and applied to signal detection in nonhomogeneous environments. Hermitian positive-definite (HPD) matrices can be used to model the sample data, while  the clutter covariance matrix is estimated by the geometric mean of a set of secondary HPD matrices. We define a projection that maps the HPD matrices in a high-dimensional manifold to a low-dimensional and  more discriminative one to increase the degree of separation of HPD matrices by maximizing the data variance.  Learning a mapping can be formulated as a two-step mini-max optimization problem in Riemannian manifolds, which can be solved by the Riemannian gradient descent algorithm. Three discriminative MIG detectors are illustrated with respect to different geometric measures, i.e., the Log-Euclidean metric, the Jensen--Bregman LogDet divergence and the symmetrized  Kullback--Leibler divergence. Simulation results  show that performance improvements of the novel MIG detectors can be achieved compared with  the conventional detectors and their state-of-the-art counterparts within nonhomogeneous environments.

\end{abstract}

%\vspace{-0.2cm}
% Note that keywords are not normally used for peerreview papers.
\begin{IEEEkeywords}
Signal detection, matrix information geometry (MIG) detectors, unsupervised learning, manifold projection, nonhomogeneous clutter.
\end{IEEEkeywords}

\IEEEpeerreviewmaketitle
%\vspace{-0.2cm}
\section{Introduction}

\IEEEPARstart{I}{mproving} the performance of signal detection in nonhomogeneous clutter is imperative in many areas, including radar \cite{4303058,5934613,7593261,7384511,9714856}, sonar \cite{7906555,1703855,8060999}, communication systems \cite{8685182}. However, the detection performance is often unsatisfactory as the number of homogeneous sample data is often  limited, not to mention the presence of interferences caused by the heterogeneity. One effective approach for enhancing the detection performance in nonhomogeneous clutter is to incorporate {\it a priori} clutter information  in designing the detectors, i.e., carrying out a knowledge-aided processing (see, e.g., \cite{6166345,6020816,7426844}), and
%Based on this assumption, significant signal-to-interference-plus-noise ratio (SINR) improvements were achieved compared with some existing CCM estimation techniques. By modeling the CCM as a multi-channel auto-regressive process with a random cross-channel (spatial) covariance structure, the authors proposed  two knowledge-aided parametric adaptive detectors in a Bayesian framework  \cite{6020816}.  Detection performance gains were shown over existing parametric solutions, especially in the case of limited data.
%Another example was given in \cite{7426844}, where a symmetrically structured power spectral density was imposed on the CCM and four adaptive decision schemes were designed in ground clutter dominated environments.
performance analysis confirmed the advantage of such an architecture over their conventional counterparts (see also \cite{5484507,4014433,5210021,4267625,5417154,7605536}). That these knowledge-aided signal detection methods can achieve significant performance improvements is owing to the sufficient information on the clutter characteristics which is often not available priorly in practical applications. Lack of knowledge about the clutter can often yield severe performance degradation.

In recent years, exploiting matrix information geometry (MIG) to deal with the problem of signal processing has been attracting extensive attention. MIG, the geometric study of matrix manifolds,  is a relative new extension of the theory of  classical information geometry, which deals with the geometric theory of probability distributions and its applications \cite{Rao1992,Chentsov1972,Efron1975Defining,97844712097,PSSY2011,LPS2008}. The reader may refer to \cite{amari2000methods,sun2016elementary} for an introduction to classical information geometry.  Many information and signal processing problems can be equivalently  transformed into discriminational problems on matrix differentiable manifolds with proper distance or divergence functions.   For instance, in \cite{7060458}, an MIG-based clutter covariance matrix (CCM) estimator was proposed in the case of limited number of sample data, and significant signal-to-interference-plus-noise ratio  gains were achieved over several standard estimators, such as the loaded sample matrix inversion. In \cite{7472241}, a new direction of arrival (DOA) estimation approach that employs geodesic distances to estimate the direction of arrival of several sources was proposed using the MIG theory. The DOA estimation was reformulated as a single-variable optimization problem on a Riemannian manifold. Simulation results showed that the proposed method improved resolution capabilities at low signal-to-noise ratio with respect to multiple signal classification and  minimum variance distortionless response. In \cite{6573681}, the problem of CCM estimation was treated as computing the geometric barycenter associated with a geometric distance for a set of secondary basic Hermitian positive-definite (HPD) matrices that yielded significant performance improvements. % were presented over the generalized-inner-product method in nonhomogeneous environments.
%In \cite{6573681,6825699}, the problem of CCM estimation was treated as computing the geometric barycenter or median associated with a geometric distance for a set of secondary basic Hermitian positive-definite (HPD) matrices, which were constructed by the secondary sample data with a condition number upper bound constraint. The geometric barycenter or median-based generalized inner product (GIP) was designed to discard secondary data containing outliers. Significant performance improvements were presented over the GIP method in nonhomogeneous clutter.
Specially, a geometric detection scheme, which we call the MIG detector, was developed by Lapuyade-Lahorgue and Barbaresco in \cite{4721049}. In MIG detectors, by  taking geometric structures of the relevant manifolds into account, {\it a priori} knowledge on the clutter characteristics is not required.

The performance of MIG detectors is closely related to the discriminative power of the utilized geometric measures, e.g., \cite{8000811,HUA2076}. In particular for HPD manifolds, the affine invariant Riemannian metric (AIRM) and the corresponding AIRM-MIG detector have been greatly studied and applied \cite{7226283,6514112,4720937,Liu2013,CO201054,9078971,6450651,7842633}. By exploiting discriminative geometric measures, it is possible to propose MIG detectors of good performances. %In \cite{8000811,HUA2076}, the authors employed several geometric measures to design MIG detectors and showed that different geometric measures possess different discriminative power and lead to different detection performances.
In \cite{HUA2017106}, they authors derived the geometric means and medians corresponding to two extended Kullback--Leibler (KL) divergences, the total KL divergence and the symmetrized KL divergence (SKLD), and particularly designed two MIG detectors based on the extended KL divergences. They were applied to target detection in K-distribution clutter, that evidenced performance gains over their state-of-the-art counterparts. In \cite{HUA2018232,huaetal2020}, the total Bregman divergence (TBD) was extended to  HPD matrix manifolds, and fortunately the geometric means derived by using the some of the mostly well-known convex functions could be derived in closed-form.  Simulation results showed that the corresponding TBD-MIG detectors outperformed the AIRM-MIG detector as well as the conventional detectors in nonhomogeneous clutter. In addition to the discriminative geometric measures mentioned above, other measures can also be defined in  HPD matrix manifolds. It is worth exploring new discriminative metrics  and designing the corresponding MIG detectors. A major limitation in these MIG detectors, nevertheless,  is that the detection performance is affected by different clutter characteristics as the discriminative power associated with a given geometric measure may change as the clutter changes.

To overcome the drawback, in this paper, we develop a projection that maps higher-order HPD matrices into a  lower-dimensional and more discriminative HPD manifold and enhances the separability of data in the unsupervised learning scenario; then we propose a type of discriminative MIG detectors, and apply them to signal detection in nonhomogeneous clutter. Main contributions of the current study are briefly summarized below.

\begin{enumerate}
  \item Inspired by the principle of principal component analysis (PCA), we propose a projection subject to an orthonormal constraint for the projection matrix, that enhances the separability between the target signal and the clutter. Learning the projection (matrix) by maximizing the variance of data becomes a two-step mini-max optimization problem in a Stiefel manifold and an HPD manifold, that can be solved by the Riemannian gradient descent (RGD) algorithm.
  Given a set of training HPD matrices that consists of two classes of data, one containing target signal and another containing only the clutter, the projection matrix can be obtained in an unsupervised way. One may consider that the PCA can only reconstruct the data in a better way but cannot lead to discrimination improvement between the samples. However, during the project, redundant information originally included in the higher-dimensional HPD matrices may be reduced during the manifold projection, leading to improvement of detection performance. From this aspect, the inspiration from PCA is rather indirect.
  \item A class of discriminative MIG detectors is designed by incorporating the manifold projection into the detection architecture. Specifically, the sample data is modeled as an HPD matrix with the diagonal loading structure, and the CCM is estimated by the geometric mean about  secondary HPD matrices. The CCM and the HPD matrix in the cell under test (CUT) are transformed into a more discriminative low-dimensional manifold. Consequently, signal detection is realized via MIG detector on a lower-dimensional HPD matrix manifold.
  \item Simulations performed in nonhomogeneous clutter verify the outperformance of the proposed discriminative MIG detectors in comparison with their state-of-the-art counterparts as well as the conventional detectors.
\end{enumerate}
%\vspace{-0.2cm}

The paper is organized as follows. The discriminative MIG detector is formulated in Section \ref{sec:pf}, and a brief introduction to MIG is given  in Section \ref{sec:MIG}. In Section \ref{sec:mp}, three geometric means are derived for the CCM estimation, and the problem of learning the projection is formulated as a two-step mini-max optimization problem in a Stiefel manifold and an HPD manifold. The performance analysis is presented in Section \ref{sec:sim}, and we conclude finally in Section \ref{sec:con}.
%\vspace{-0.1cm}

\textit{Notations:} We use boldface lowercase  (uppercase)  letters  to denote vectors (matrices). Matrix (or vector) transpose and conjugate transpose are denoted by the superscripts $(\cdot)^{\operatorname{T}}$ and $(\cdot)^{\operatorname{H}}$, respectively. Determinant and trace of a matrix are respectively denoted by $\operatorname{det}(\cdot)$  and $\operatorname{tr}(\cdot)$. The $N\times N$ identity matrix is denoted by $\bm{I}_N$ or simply $\bm{I}$. %The Frobenius norm of a matrix with respect to the Frobenius metric is simply denoted by $\norm{ \cdot }$.
The notations $\mathbb{C}^N$ and $\mathbb{C}^{M\times N}$ represent the set of $n$-dimensional complex vectors and $M \times N$ complex matrices, respectively. The imaginary unit is $\operatorname{i}$, and  finally $\operatorname{E}[\cdot]$ denotes the statistical expectation.
%\vspace{-0.1cm}

\section{Problem Formulation}
\label{sec:pf}
Let $\bm{x} = [x_0, x_1, \ldots, x_{N-1}]^{\operatorname{T}}$ be the sample data collected from $N$ (temporal, spatial, or spatial-temporal) channels. In general, the problem of signal detection is interpreted as the following  binary hypothesis testing
\begin{equation} \label{eq::hypothesis_testing}
\left\{
\begin{aligned}
&\mathcal{H}_0: \left\{
\begin{aligned}
&\bm{x} = \bm{c}, \\
&\bm{x}_k = \bm{c}_k, \quad  k \in [K],% k=1,2,\ldots,K,
\end{aligned} \right.\\
&\mathcal{H}_1: \left\{
\begin{aligned}
&\bm{x} = \alpha \bm{p} + \bm{c}, \\
&\bm{x}_k = \bm{c}_k, \quad k \in [K],% k=1,2,\ldots,K,
\end{aligned} \right. \\
\end{aligned} \right.
\end{equation}
where $[K]$ denotes the set of $\{ 1,2,\ldots,K \}$ with $K$ the number of secondary data, $\bm{c}$ and $\bm{c}_k$ are the clutter data, and $\bm{x}$ and $\bm{x}_k$ denote the observation data. In particular, $\bm{x}$ and $\bm{c}$ represent data of the CUT. Here, $\mathcal{H}_0$ and $\mathcal{H}_1$ denote the null and alternative hypotheses that correspond to the absence and presence of a target signal, respectively.  The unknown  complex parameter $\alpha$ is relevant to the channel propagation effects and target reflectivity. The known steering vector $\bm{p}$ is given by
\begin{equation}
\bm{p} = \frac{1}{\sqrt{N}}\left[1, \exp\left(-\operatorname{i}2\pi f_d\right), \ldots, \exp\left(-\operatorname{i}2\pi f_d(N-1)\right)\right]^{\operatorname{T}},
\end{equation}
where $f_d$ denotes the normalized Doppler frequency.

The observation data is assumed to obey a multivariate complex Gaussian distribution with zero mean. Therefore, statistical information of the sample data is closely related to  the covariance matrix.
The power or correlation of the sample data, which can be represented as an HPD matrix,  is employed for distinguishing the target signal from the clutter. Various structures can be specified  for the HPD matrix, for instance,  the Toeplitz structure \cite{7060458}, the diagonal loading \cite{8450037}, the shrinkage estimators \cite{8052571} and the persymmetric covariance estimators \cite{8496824}. The HPD manifolds subject to different matrix structures possess different geometric structures. The resulting differences in the detection performance cased by different matrix structures will be analyzed separately. Here, we exploit the HPD matrix with the diagonal loading structure to model the sample data. The diagonal loading structure has been successfully applied in signal detection \cite{4383590,7181,5417174}.
%For instance, diagonal loading can be used to reduce the main-lobe distortion maintaining also lower sidelobes to stabilize the beampattern, as the large eigenvalues of the CCM due to strong interference are not significantly affected by the loading process, whereas the smaller eigenvalues are increased.
The diagonal loading HPD matrix can be expressed by adding diagonal matrix to the sample covariance matrix (SCM), i.e.,
\begin{equation} \label{eq::diagonal_loading}
\bm{R} = \bm{r}\bm{r}^{\operatorname{H}} + \operatorname{tr}\left(\bm{r}\bm{r}^{\operatorname{H}}\right)\bm{I},
\end{equation}
where $\bm{r} = [r_0, r_1, \dots, r_{N-1}]^{\operatorname{T}}$ denotes correlation of the sample data, namely
\begin{equation}
r_l = \operatorname{E}[x_i \overline{x}_{i+l}],\quad  0\leq l \leq N-1, 1\leq i \leq N - l - 1.
\end{equation}
where $\overline{x}_i$ denotes the conjugate of $x_i$. Ergodicity of stationary Gaussian process allows us to approximate $r_l$  by the following estimator
\begin{equation}
\widetilde{r}_l = \frac{1}{N} \sum_{i=0}^{N-1- l } {x_i\overline{x}_{i+l}},\quad  0 \leq l \leq N-1.
\end{equation}

Using the diagonal loading formalism \eqref{eq::diagonal_loading}, each sample data can be represented as an HPD matrix with a diagonal loading structure as the new observation. The set of all $N\times N$ HPD matrices forms a differentiable manifold; see Section \ref{sec:MIG} for more details.  Assuming that $K$ secondary HPD matrices $\{  \bm{R}_k \}_{ k \in [K]}$ are available, we employ the geometric mean $\bm{R}_{\mathcal{G}} = \mathcal{G}(\bm{R}_1, \bm{R}_2, \ldots, \bm{R}_K)$ to estimate the CCM. From the viewpoint of MIG, the problem of binary hypothesis testing  Eq. \eqref{eq::hypothesis_testing} can be rewritten as (see e.g., \cite{huaetal2020}),
\begin{equation}
\left\{
\begin{aligned}
\mathcal{H}_0: \bm{R}=\bm{R}_\mathcal{G}, \\
\mathcal{H}_1: \bm{R}\neq\bm{R}_\mathcal{G}.
\end{aligned}
\right.
\end{equation}

Given the observation $\bm{R}_D$ and the CCM estimate $\bm{R}_{\mathcal{G}}$, by utilizing geometric structure of HPD matrix manifolds,  signal detection can be interpreted as the discrimination of two HPD matrices $\bm{R}_D$ and $\bm{R}_{\mathcal{G}}$ in a differentiable manifold.  Let us consider the null hypothesis $\mathcal{H}_0: \bm{R}=\bm{R}_\mathcal{G}$ versus the alternative hypothesis $\mathcal{H}_1: \bm{R}\neq\bm{R}_\mathcal{G}$ based on a set of observations $\{\bm{R}_k \}_{k \in [K]}$.
% $\{\bm{R}_1, \bm{R}_2, \ldots, \bm{R}_K \}$.
The problem of signal detection can be understood as to determine the inner of isosurfaces of the HPD matrix manifold determined by a distance or divergence, as illustrated in Fig. \ref{Detection_Scheme}. The hypothesis $\mathcal{H}_0$ is rejected if the observation $\bm{R}_D$ of CUT does not belong to the inner of an isosurface.

%\vspace{-0.2cm}
\begin{figure}[htbp]
\centering
\includegraphics[width=8.2cm,angle=0]{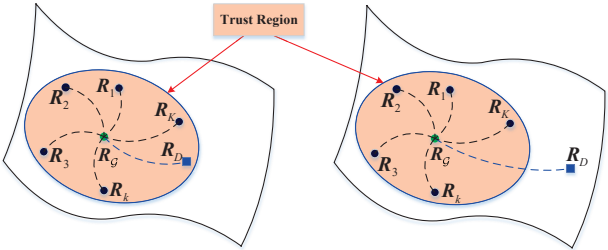}
\caption{The diagram for the geometric interpretation of signal detection}
\label{Detection_Scheme}
%\vspace{-0.5cm}
\end{figure}
%\vspace{-0.2cm}

In our detection framework, we learn a projection matrix $\bm{W} \in \mathbb{C}^{N \times M} (M \leq N)$ of full rank to maximize the variance of the data by resorting to the training HPD matrices in an unsupervised scenario, where any $N\times N$ HPD matrix $\bm{R}$ will be mapped into a more discriminative low-dimensional manifold by
\begin{equation}
f_{\bm{W}}(\bm{R}) = \bm{W}^{\operatorname{H}}\bm{R}\bm{W} \in \mathbb{C}^{M \times M},
\end{equation}
where $\bm{W}$ is conventionally assumed in the (compact and complex) Stiefel manifold
\begin{equation}
\operatorname{St}(M,\mathbb{C}^N)=\left\{\bm{A}\in\mathbb{C}^{N\times M} \mid \bm{A}^{\operatorname{H}}\bm{A}=\bm{I}_M\right\}.
\end{equation}
Consequently, the problem of signal detection becomes to determining the inner of an isosurface in an HPD matrix manifold associated to  a given distance or divergence (see Fig. \ref{Detection_Scheme}), namely
\begin{equation}
d\left(f_{\bm{W}}(\bm{R}_\mathcal{G}), f_{\bm{W}}(\bm{R}_D)\right) \mathop{\gtrless}\limits_{\mathcal{H}_0}^{\mathcal{H}_1} \gamma,
\end{equation}
where $d(\cdot,\cdot)$ is the distance or divergence that used to measure the dissimilarity between two points on the HPD matrix manifold. The hypothesis  $\mathcal{H}_1$, meaning the presence of a target signal,  is accepted if the observation $f_{\bm{W}}(\bm{R}_D)$ of CUT lies outside of an isosurface centered at $f_{\bm{W}}(\bm{R}_\mathcal{G})$ with radius $\gamma$, which is the detection threshold.

%It is clear that the statistic of MIG detectors is the dissimilarity between the HPD matrix of CUT and the CCM derived by a geometric measure. Different geometric measures reflect different local geometric structures of HPD manifolds and yield different discriminative power. Then, MIG detectors based on different geometric measures have different detection performances. In this paper, we propose a projection that maps HPD matrices into a more discriminative low-dimensional manifold to increase the separability of data, and consequently several improvements in the detection performance can be achieved.

\section{Preliminaries of Matrix Information Geometry}
\label{sec:MIG}
Before moving to an MIG solution to the problem formulated in Section \ref{sec:pf} above, we briefly review the theory of MIG that is relevant to the current study in this section.

\subsection{HPD manifolds}
The general linear group $GL(N,\mathbb{F})$ consists of all $N\times N$ invertible matrices  with $\mathbb{F}$  either real $\mathbb{R}$ or complex $\mathbb{C}$. The Frobenius  metric\footnote{ It is also called the Hilbert--Schmidt inner product.} is defined by
\begin{equation}\label{eq:glmetric}
\langle \bm{X},\bm{Y}\rangle:=\operatorname{tr}(\bm{X}^{\operatorname{H}}\bm{Y}),\quad \bm{X},\bm{Y}\in GL(N,\mathbb{F}).
\end{equation}
%where $\bm{X}^{\operatorname{H}}$ is the conjugate transpose of $\bm{X}$. In the real case, it is simply the transpose of a matrix $\bm{X}$. %Subgroups and submanifolds of $GL(N,\mathbb{F})$ can admit topological or geometric properties.
%In many cases, such as subgroups of $GL(N,\mathbb{F})$ with better geometric or topological properties, one may define various metrics and even dual connections as have been greatly investigated for statistical models in classical information geometry.
In the current paper, our main interest is HPD matrices, that form a subspace of  $GL(N,\mathbb{C})$.

%Denote $\mathscr{H}(N,\mathbb{C})$ as the set of $N \times N$ Hermitian matrices:
%\begin{equation}
%\mathscr{H}(N,\mathbb{C}) = \left\{ \bm{A} \mid \bm{A} = \bm{A}^{\operatorname{H}}, \bm{A}\in \mathbb{C}^{N\times N} \right\}.
%\end{equation}
%A Hermitian matrix $\bm{A} \in \mathscr{H}(N,\mathbb{C})$ is positive-definite, expressed as $\bm{A}\succ \bm{0}$, if the quadratic expression $\bm{x}^{\operatorname{H}}\bm{A}\bm{x} > 0$ holds for all $ \bm{x} \in \mathbb{C}^N, \bm{x} \neq {\bf 0}$.
The set of $N\times N$ HPD matrices is denoted by $\mathscr{P}(N,\mathbb{C})$ which is a subset of $GL(N,\mathbb{C})$ and naturally a differentiable manifold. %denoted by $\mathscr{P}(n,\mathbb{C})$ be the set of HPD matrices, it can be noted by
Each element $\bm{A}\in\mathscr{P}(N,\mathbb{C})$ is Hermitian and positive-definite, that is
\begin{equation}
\bm{A}^{\operatorname{H}}=\bm{A} \text{ and }   \bm{x}^{\operatorname{H}}\bm{A}\bm{x} > 0 \text{ for all } \bm{0}\neq \bm{x}\in\mathbb{C}^N.
\end{equation}
%which is naturally a differentiable manifold.
%It is well known that an HPD matrix naturally embeds in a underlying matrix manifold that accounts for the curvature of the non-linear space $\mathscr{P}(n,\mathbb{C})$.
%Each element of the matrix manifold $\mathscr{P}(N,\mathbb{C})$ denotes an HPD matrix.
The difference of two HPD matrices can be evaluated by a distance, a divergence or other measures  defined in  $\mathscr{P}(N,\mathbb{C})$. It is crucial to specify these measures properly in applications, as different measures  will lead to different isosurfaces.

\subsection{Riemannian structures of HPD manifolds}
Except the induced subspace Frobenius  metric, the space $\mathscr{P}(N,\mathbb{C})$ is a Riemannian manifold  with the AIRM
\begin{equation}
\langle \bm{A},\bm{B} \rangle_{\bm{P}} := \operatorname{tr}\left(\bm{P}^{-1}\bm{A}\bm{P}^{-1}\bm{B}\right),\quad \bm{A},\bm{B}\in T_{\bm{P}}\mathscr{P}(N,\mathbb{C}). % \langle \bm{P}^{-1/2}\bm{A}\bm{P}^{-1/2}, \bm{P}^{-1/2}\bm{B}\bm{P}^{-1/2} \rangle.
\end{equation}
Under the AIRM, its curvature is non-positive
 %which has been extensively studied during the recent few decades
  \cite{bridson2013metric,8624413,8060999,8920217}.  In the following, we are going to introduce some of the mostly well-known geometric measures in the manifold $\mathscr{P}(N,\mathbb{C})$ as either a metric space equipped with the Frobenius metric or a Riemannian manifold.

In the Riemannian manifold $\mathscr{P}(N,\mathbb{C})$,  exponential map and logarithm map can naturally be defined on the tangent bundle $T\mathscr{P}(N,\mathbb{C})=\bigcup\limits_{\bm{P}} T_{\bm{P}}\mathscr{P}(N,\mathbb{C})$ by  using the geodesics. They are related to matrix exponentials and matrix logarithms.   Matrix exponential for a general matrix $\bm{X}$  is defined by a Taylor series
\begin{equation}
\exp(\bm{X}) = \sum_{i=0}^{+\infty} \frac{\bm{X}^i}{i!} .
\end{equation}
Logarithm of a matrix is defined as the inversion of matrix exponential. Unfortunately, it is not always well-defined as a function.
The following lemma defines the principle logarithm of an invertible matrix together with some of its important properties, which will be used later.

% and its tangent space $T_{\bm{P}}\mathscr{P}(n,\mathbb{C})$ can be mapped to each other by the principal logarithm and exponential operators.

\begin{lem}[\cite{Hig2008,Moa2005}] \label{lem:aa}
Let $\bm{X}$ be an invertible matrix and assume that none of its eigenvalues lie in the closed negative real line. Then, there exists a unique matrix logarithm of $\bm{X}$ whose eigenvalues lie in the strip
\begin{equation*}
\{z\in \mathbb{C}\mid -\pi<\operatorname{Im}(z)<\pi\}.
\end{equation*}
It is referred to as the principle logarithm and denoted by $\operatorname{Log}\bm{X}$.

The principle (matrix) logarithm satisfies the following properties.
\begin{itemize}
\item[(i)] Each pair of the matrices $[(\bm{X}-\bm{I})s+\bm{I}]^{-1}$,  $\bm{X}$ and $\operatorname{Log}\bm{X}$   commutes any real number $s$.
\item[(ii)]   The following matrix integral is valid:
\begin{equation*}
\begin{aligned}
&\int_0^1[(\bm{X}-\bm{I})s+\bm{I}]^{-2}\operatorname{d}\!s\\
&\quad\quad\quad =(\bm{I}-\bm{X})^{-1}[(\bm{X}-\bm{I})s+\bm{I}]^{-1}\Big|_{s=0}^1\\
&\quad\quad\quad =\bm{X}^{-1}.
\end{aligned}
\end{equation*}
\item[(iii)] Let $\bm{A}(\varepsilon)$ be an invertible matrix satisfying the unique principal logarithm existence condition above. Furthermore, assume  $\bm{A}(\varepsilon)$  depends on the  real parameter $\varepsilon$ smoothly. Then we have
\begin{equation*}
\begin{aligned}
\frac{\operatorname{d}}{\operatorname{d}\!\varepsilon}\operatorname{Log}\bm{A}(\varepsilon)&=\int_0^1
[(\bm{A}(\varepsilon)-\bm{I})s+\bm{I}]^{-1}  \frac{\operatorname{d}}{\operatorname{d}\!\varepsilon}\bm{A}(\varepsilon)\\
&\quad\quad \quad \times [(\bm{A}(\varepsilon)-\bm{I})s+\bm{I}]^{-1}\operatorname{d}\!s.
\end{aligned}
\end{equation*}
\end{itemize}
\end{lem}

%The principal logarithm $\operatorname{Log}(\cdot): \mathscr{P}(n,\mathbb{C}) \rightarrow T_{\bm{P}}\mathscr{P}(n,\mathbb{C})$ is defined as
%\begin{equation}
%\operatorname{Log}(\bm{X}) = \sum_{i=1}^\infty \frac{(-1)^{i-1}}{i} (\bm{X}-\bm{I})^i.
%\end{equation}

In the Riemannian manifold $\mathscr{P}(N,\mathbb{C})$ equipped with the AIRM, the distance of two points $\bm{X}, \bm{Y} \in \mathscr{P}(N,\mathbb{C})$ is given by the  length of the local geodesic with them as the endpoints, reading
\begin{equation}\label{eq:AIRM}
\begin{aligned}
d_R^2(\bm{X},\bm{Y}) &= \norm{ \operatorname{Log}\left(\bm{X}^{-1/2}\bm{Y}\bm{X}^{-1/2}\right) }^2\\
& = \sum_{i=1}^N \ln^2\lambda_i,
\end{aligned}
\end{equation}
where $\lambda_1,\lambda_2,\ldots,\lambda_N$  are the eigenvalues of the matrix $\bm{X}^{-1/2}\bm{Y}\bm{X}^{-1/2}$. We use  $\norm{\cdot}$ to denote the Frobenius norm $\norm{\bm{A}}^2=\operatorname{tr}(\bm{A}^{\operatorname{H}}\bm{A})$ of a matrix $\bm{A}$, induced from the Frobenius metric \eqref{eq:glmetric}.

Unfortunately,  the computational cost of the AIRM distance is often expensive in practical applications. An alternative choice, the Log-Euclidean metric (LEM) \cite{637996}, is defined as follows
\begin{equation}
\langle \bm{A},\bm{B} \rangle_{\bm{P}}^{\operatorname{LE}} := \langle D_{\bm{A}}\operatorname{Log}\bm{P}, D_{\bm{B}}\operatorname{Log}\bm{P} \rangle,
\end{equation}
where $\bm{A},\bm{B}\in T_{\bm{P}}\mathscr{P}(N,\mathbb{C})$ and $D_{\bm{A}}\operatorname{Log}\bm{P}$ denotes the directional derivative of the matrix logarithm along a tangent vector $\bm{A}$ at a point $\bm{P}$. %The LEM metric is more efficient in the sense of computational cost, compared with the AIRM metric.
The LEM distance of two HPD matrices $\bm{X}, \bm{Y}\in \mathscr{P}(N,\mathbb{C})$ is given by the length of the local geodesic as
\begin{equation}\label{D:LEM}
d_L^2(\bm{X},\bm{Y}) = \norm{\operatorname{Log}\bm{X}-\operatorname{Log}\bm{Y}}^2.
\end{equation}

\subsection{Divergences of HPD matrices}
By viewing the differentiable manifold $\mathscr{P}(N,\mathbb{C})$ as a metric space equipped with the Frobenius metric, many other geometric measures can also be defined. We will be focused on the the Jensen--Bregman LogDet divergence (JBLD) \cite{6378374} and the symmetrized  Kullback--Leibler divergence (SKLD)   \cite{HUA2017106} in the current paper.  The JBLD and SKLD of two HPD matrices $\bm{X}, \bm{Y} \in \mathscr{P}(N,\mathbb{C})$ are respectively gived by
% \begin{defn}
%The JBLD  between two HPD matrices $\bm{X}, \bm{Y} \in \mathscr{P}(N,\mathbb{C})$ is defined as
\begin{equation}\label{D:JBLD}
d_J^2(\bm{X},\bm{Y}) =\ln \operatorname{det}\left(\frac{\bm{X}+\bm{Y}}{2}\right) - \frac{1}{2}\ln\operatorname{det}(\bm{X}\bm{Y})
\end{equation}
and
%\end{defn}
%
%\begin{defn}
%The SKLD between two HPD matrices $\bm{X}, \bm{Y} \in \mathscr{P}(N,\mathbb{C})$ is defined as
\begin{equation}\label{D:SKLD}
d_S^2(\bm{X},\bm{Y}) = \frac{1}{2}\operatorname{tr}\left( \bm{Y}^{-1}\bm{X} + \bm{X}^{-1}\bm{Y} - 2\bm{I}\right).
\end{equation}
%\end{defn}

%Note that the AIRM and LEM metrics lead to a true geodesic distance on the manifold $\mathscr{P}(N,\mathbb{C})$.
Note that among all geometric measures introduced above, the AIRM, the JBLD and the SKLD are invariant with respect to affine transformations.

In the study of optimization problems in $\mathscr{P}(N,\mathbb{C})$, we often need to compute the gradient of a function $F(\bm{R})$, which is defined by the  covariant/directional derivative associated to a given metric, e.g., a Riemannian metric or simply the Frobenius metric, as follows
\begin{equation}\label{def:gra0}
\langle \nabla F(\bm{R}),\bm{A}\rangle:=\frac{\operatorname{d}}{\operatorname{d}\!\varepsilon}\Big|_{\varepsilon=0} F(\gamma(\varepsilon)),  \quad \forall \bm{A}\in T_{\bm{R}}\mathscr{P}(N,\mathbb{C}),
\end{equation}
where $\gamma:[0,1]\rightarrow \mathscr{P}(N,\mathbb{C})$ is the unique local curve satisfying $\gamma(0)=\bm{R}$ and $\dot{\gamma}(0)=\bm{A}$. By taking the linear part into account, it can  be rewritten as
\begin{equation}\label{def:gra}
\langle \nabla F(\bm{R}),\bm{A}\rangle:=\frac{\operatorname{d}}{\operatorname{d}\!\varepsilon}\Big|_{\varepsilon=0} F(\bm{R}+\varepsilon\bm{A}), \quad \forall \bm{A}\in T_{\bm{R}}\mathscr{P}(N,\mathbb{C}).
\end{equation}

\section{Geometric Means and Unsupervised Manifold Projection}
\label{sec:mp}
\subsection{Geometric Means}
It is well known that  the arithmetic mean of  a set of $K$ positive real numbers $\{x_k \}_{k \in [K]}$
% $\{x_1, x_2, \ldots, x_K \}$
can be calculated by
\begin{equation}
\widehat{x} = \frac{1}{K}\sum_{k=1}^K x_k.
\end{equation}
 In fact, the arithmetic mean is the minimum value of the sum of the squares, namely
\begin{equation}
\widehat{x} := \underset{x \in \mathbb{R}^+}{{\arg\min}} \sum_{k=1}^K | x - x_k |^2,
\end{equation}
where $| x - x_k |$ denotes  the distance between $x$ and $x_k$. Geometric mean of a set of HPD matrices can similarly be defined.

\begin{defn}
Given a set of $K$ HPD matrices $\{\bm{R}_k \}_{k \in [K]}$,
% $\{ \bm{R}_1,\bm{R}_2,...,\bm{R}_K \}$,
the geometric mean with respect to a geometric measure $d: \mathscr{P}(N,C) \times \mathscr{P}(N,C) \rightarrow \mathbb{R}$ is obtained through the following optimization problem
\begin{equation}
\bm{\widehat{R}} := \underset{\bm{R} \in \mathscr{P}(N,\mathbb{C})}{{\arg\min}} \sum_{k=1}^K d^2\left(\bm{R}_k,\bm{R}\right).
\label{eq:def_gm}
\end{equation}
\end{defn}

Geometric means of a set of HPD matrices can not always be calculated in closed form; alternatively, the fixed-point iteration has proven to be  effective  for calculating them numerically, e.g., \cite{637996,104240020,8000811}. In the below, we summarize the algorithms or analytic expressions for computing the geometric means corresponding to the three measures introduced above, i.e.,  the LEM distance \eqref{D:LEM}, the AIRM geodesic distance \eqref{eq:AIRM}, the JBLD \eqref{D:JBLD} and the SKLD \eqref{D:SKLD}.

\begin{prop}\label{prop:LEMm}
The LEM mean of HPD matrices $\{\bm{R}_k \}_{k \in [K]}$
% $\{\bm{R}_1, \bm{R}_2, \ldots, \bm{R}_K \}$
is given by \cite{637996}
\begin{equation}\label{eq:LEM_Mean}
\widehat{\bm{R}} = \operatorname{exp}\bigg(\frac{1}{K}\sum_{k=1}^K \operatorname{Log}\bm{R}_k \bigg).
\end{equation}
\end{prop}

%\begin{proof}
%A proof is available in .
%\end{proof}

\begin{prop}\label{prop:AIRMm}
The AIRM mean of $\{\bm{R}_k \}_{k \in [K]}$
% $\{\bm{R}_1, \bm{R}_2, \ldots, \bm{R}_K \}$
is determined by \cite{Moa2005}
\begin{equation}
\sum_{k=1}^K\operatorname{Log}\left(\bm{R}_k^{-1}\bm{\widehat{R}}\right)=0,
\end{equation}
which can be obtained using the following fixed-point iteration \cite{Moa2006}:
\begin{equation}
\begin{aligned}
\widehat{\bm{R}}_{t+1}&=a \widehat{\bm{R}}_t\\
&+(a-1)\sum_{k=2}^K\operatorname{Log}\left(\exp\left(\frac{\widehat{\bm{R}}_t}{2}\right){\bm{R}}_k^{-1}\exp\left(\frac{\widehat{\bm{R}}_t}{2}\right)\right),
\end{aligned}
\end{equation}
where $1-1/K<a<1$, $t$ denotes the iterative index, and the initial value is
\begin{equation}
\widehat{\bm{R}}_0=\frac{1}{K}\sum_{k=1}^K\operatorname{Log}\bm{R}_k.
\end{equation}
\end{prop}

\begin{prop}\label{prop:JBLDm}
The JBLD mean of HPD matrices $\{\bm{R}_k \}_{k \in [K]}$
% $\{\bm{R}_1, \bm{R}_2, \ldots, \bm{R}_K \}$
can be obtained through the fixed-point iteration \cite{104240020,8000811}:
\begin{equation}\label{eq:JBLD_Mean}
\bm{\widehat{R}}_{t+1} = \Bigg(\frac{1}{K}\sum_{k=1}^K  \bigg(\frac{\bm{\widehat{R}}_t + \bm{R}_k}{2}\bigg)^{-1} \Bigg)^{-1}.
\end{equation}
\end{prop}

%\begin{proof}
%See \cite{104240020,8000811}.
%\end{proof}

\begin{prop}\label{prop:SKLDm}
The SKLD mean of HPD matrices $\{\bm{R}_k \}_{k \in [K]}$
% $\{\bm{R}_1, \bm{R}_2, \ldots, \bm{R}_K \}$
is
\begin{equation}\label{eq:SKL_Mean}
\bm{\widehat{R}} = \bm{A}^{-1/2}\left( \bm{A}^{1/2}\bm{B}\bm{A}^{1/2} \right)^{1/2}\bm{A}^{-1/2},
\end{equation}
where
\begin{equation}
\bm{A} = \sum_{k=1}^K \bm{R}_k^{-1}, \quad \bm{B} = \sum_{k=1}^K \bm{R}_k.
\end{equation}
\end{prop}

\begin{proof}
A proof is provided in Appendix \ref{app:SKLD}; see also \cite{HPL2021}.
\end{proof}

\subsection{Unsupervised Manifold Projection}

In this subsection, we introduce the manifold projection that maps HPD matrices from a high-dimensional manifold to a more discriminative lower-dimensional one by maximizing the variance of data.

Recall that the variance of a set of vectors $\{\bm{x}_k \}_{k \in [K]}$
% $\{\bm{x}_1, \bm{x}_2, \ldots, \bm{x}_K\}$
in a Euclidean space is  given by
\begin{equation}
\operatorname{Var}= \frac{1}{K}\sum_{k=1}^K \norm{\bm{x}_k - \bm{\widehat{x}}}_2^2 \quad \text{with} \quad \bm{\widehat{x}} = \frac{1}{K}\sum_{k=1}^K \bm{x}_k,
\end{equation}
where $\norm{\cdot}_2$ denotes the $l_2$ norm, and $\bm{\widehat{x}}$ is the mean of the set of vectors.

Given a set of HPD matrices $\{\bm{R}_i \}_{i \in [J+K]}$
% $\{ \bm{R}_1, \bm{R}_2, \ldots, \bm{R}_{J+K} \}$
that contains $J$ CCMs and $K$ HPD matrices with a target signal, the variance can similarly be defined as
\begin{equation}\label{eq:var}
\operatorname{Var}= \frac{1}{J+K}\sum_{i=1}^{J+K} d^2\left(\bm{R}_i, \bm{\widehat{R}}\right).
\end{equation}
where $\bm{\widehat{R}}$ denotes the mean of $J+K$ HPD matrices, which can be derived using  Eq. \eqref{eq:def_gm} with respect to the LEM distance, the JBLD or the SKLD. Note that the variance \eqref{eq:var}  can be interpreted as a deterministic counterpart of the variance function of  a probability distribution defined in the HPD manifold \cite{7819552}.

As briefly introduced in Section \ref{sec:pf},  we propose a manifold projection that maps HPD matrices into a more lower-dimensional manifold  maximizing the data variance. The projection is defined as
\begin{equation}
\begin{aligned}
f_{\bm{W}}:\mathscr{P}(N,\mathbb{C})&\rightarrow \mathscr{P}(M,\mathbb{C})\\
\bm{R}&\mapsto \bm{W}^{\operatorname{H}}\bm{R}\bm{W},
\end{aligned}
\end{equation}
where $M\leq N$ and $\bm{W}\in \operatorname{St}(M,\mathbb{C}^N)\subset \mathbb{C}^{N\times M}$. Obviously, $\bm{W}$ is of maximal rank and $\bm{W}^{\operatorname{H}}\bm{W}=\bm{I}_M$.
%Specifically, for a projection matrix $\bm{W}$, we can transform an HPD matrix $\bm{R} \in \mathbb{C}^{N \times N}$ as a lower-dimensional on $\bm{W}^{\operatorname{H}}\bm{R}\bm{W}, \bm{W} \in \mathbb{C}^{N \times M}, M < N$.
Therefore, for a set of HPD matrices $\{\bm{R}_i \}_{i \in [J+K]}$
% $\{ \bm{R}_1, \bm{R}_2, \ldots, \bm{R}_{J+K} \}$
in $\mathscr{P}(N,\mathbb{C})$, learning a mapping to achieve maximal variance is equivalent to searching a projection matrix $\bm{W}$ in the Stiefel manifold. Namely, the problem becomes solving the optimization problem
\begin{equation}\label{eq:maxeq}
\begin{aligned}
\overline{\bm{W}} &:= \underset{\bm{W}\in \operatorname{St}(M,\mathbb{C}^N) }{\arg \max} \frac{1}{J+K} \sum_{i=1}^{J+K} d^2\left(f_{\bm{W}}(\bm{R}_i), \bm{\widehat{Z}} \right)\\
& ~ =\underset{\bm{W}\in \operatorname{St}(M,\mathbb{C}^N) }{\arg \max} \frac{1}{J+K} \sum_{i=1}^{J+K} d^2\left(\bm{W}^{\operatorname{H}}\bm{R}_i\bm{W}, \bm{\widehat{Z}} \right),
\end{aligned}
\end{equation}
where $\bm{\widehat{Z}}$ is the geometric mean of the set $\{\bm{W}^{\operatorname{H}}\bm{R}_i\bm{W}\}_{i \in [J+K]}$ in $\mathscr{P}(M,\mathbb{C})$, namely
\begin{equation}
\bm{\widehat{Z}} = \underset{\bm{Z} \in \mathscr{P}(M,\mathbb{C})}{{\arg\min}} \sum_{i=1}^{J+K} d^2\left(\bm{W}^{\operatorname{H}}\bm{R}_i\bm{W},\bm{Z}\right).
\end{equation}

%we model the mapping as a transform function $f_{\bm{W}}(\bm{R}) = \bm{W}^{\operatorname{H}}\bm{R}\bm{W}$ with an orthonormal constraint, where $\bm{R} \in \mathscr{P}(N,\mathbb{C}), \bm{W} \in \mathbb{C}^{N \times M}, M < N, \bm{W}^{\operatorname{H}}\bm{W} = \bm{I}$. Then, learning a mapping is formulated as a two-step mini-max optimization problem on the Stiefel manifold.
%% which can be solved by the {\color{red} Riemannian gradient descent method}.
%Finally, three discriminative projections related to the LEM, JBLD and SKLD are derived.

\begin{rem}\label{rem8}
Solving the optimization problem \eqref{eq:maxeq} is a very complex and nonlinear problem as $\widehat{\bm{Z}}\in\mathscr{P}(M,\mathbb{C})$ also depends on $\bm{W}$. Although for the LEM distance and the SKLD, we can  obtain the means (see Propositions \ref{prop:LEMm} and \ref{prop:SKLDm}) respectively as
\begin{equation}
\widehat{\bm{Z}}_L=\exp\left(\frac{1}{J+K}\sum_{i=1}^{J+K}\operatorname{Log}\left(\bm{W}^{\operatorname{H}}\bm{R}_i\bm{W}\right)\right)
\end{equation}
and
\begin{equation}
\widehat{\bm{Z}}_S=\bm{A}^{-1/2}\left( \bm{A}^{1/2}\bm{B}\bm{A}^{1/2} \right)^{1/2}\bm{A}^{-1/2},
\end{equation}
where
\begin{equation}
\bm{A} = \sum_{i=1}^{J+K}\left(\bm{W}^{\operatorname{H}} \bm{R}_i\bm{W}\right)^{-1}, \quad \bm{B} = \sum_{i=1}^{J+K} \bm{W}^{\operatorname{H}}\bm{R}_i\bm{W}.
\end{equation}
\end{rem}

Since it is difficult to solve the projection matrix  $\overline{\bm{W}}$  from \eqref{eq:maxeq} in closed-form, we formulate it as a two-step mini-max optimization problem as follows:
%Consider the difference between the geometric mean $\bm{\widetilde{X}}$ of the compress set $\bm{W}^{\operatorname{H}}\bm{X}_i\bm{W}, i=1,2, \dots, m+n$ and the compress mean $\bm{W}^{\operatorname{H}}\bm{\widehat{X}}\bm{W}$ of the original set $\bm{X}_i, i=1,2, \dots, m+n$, where $\bm{W}^{\operatorname{H}}\bm{W} = \bm{I}$, then, learning a mapping from the high-dimensional HPD matrix manifold to a more discriminative low-dimensional one can be formulated as a two-step mini-max optimization problem:
\begin{equation}\label{eq:opt_problem}
\begin{aligned}
\overline{\bm{W}}_{t+1} &= \underset{\bm{W}\in \operatorname{St}(M,\mathbb{C}^N) }{{\arg\max}} \frac{1}{J+K} \sum_{i=1}^{J+K} d^2\left(\bm{W}^{\operatorname{H}}\bm{R}_i\bm{W}, \bm{\widehat{Z}}_t \right), \\
\bm{\widehat{Z}}_{t+1} &= \underset{\bm{Z} \in \mathscr{P}(M,\mathbb{C})}{{\arg\min}} \sum_{i=1}^{J+K} d^2\left(\overline{\bm{W}}_{t+1}^{\operatorname{H}}\bm{R}_i\overline{\bm{W}}_{t+1},\bm{Z}\right), \\
\end{aligned}
\end{equation}
where  $t$ denotes the iterative step. At each step, the minimal problem for $\widehat{\bm{Z}}$ can either be solved analytically by Remark \ref{rem8} or numerically by Proposition \ref{prop:JBLDm}; the maximal problem can be transformed into a minimal problem
\begin{equation}
\underset{\bm{W}\in \operatorname{St}(M,\mathbb{C}^N) }{{\arg\min}} \psi(\bm{W}),
\end{equation}
that can be solved % {\color{red} by reformulating it as a minimal problem $-\underset{\bm{W}\in \operatorname{St}(M,\mathbb{C}^N) }{{\arg\max}} \psi(\bm{W})$}
by the RGD algorithm \eqref{RGD} (see also\cite{Smith1993,Udr1994}), where at each step $t$,
\begin{equation}\label{eq:psi}
\psi(\bm{W})= - \frac{1}{J+K} \sum_{i=1}^{J+K} d^2\left(\bm{W}^{\operatorname{H}}\bm{R}_i\bm{W}, \bm{\widehat{Z}}_t \right).
\end{equation}
%The stop criterion is given as $\norm{\bm{W}_{t+1}-\bm{W}_t}<\epsilon$, where $\epsilon$ is a small positive number.

The Riemannian gradient of  a function $\psi(\bm{W})$ defined on the Stiefel manifold $\operatorname{St}(M,\mathbb{C}^N)$ is given by \cite{AMS2008}
\begin{equation}\label{eq:Rgrad}
\operatorname{grad}\psi(\bm{W}) = \nabla \psi(\bm{W})-\bm{W}\times \operatorname{sym}\left(\bm{W}^{\operatorname{H}}\nabla \psi(\bm{W})\right),
\end{equation}
where
\begin{equation}
\operatorname{sym}(\bm{A})=\frac{\bm{A}+\bm{A}^{\operatorname{H}}}{2}
\end{equation}
denotes the symmetric part of a matrix $\bm{A}$, and $\nabla\psi(\bm{W})$ is the Euclidean gradient induced from the Frobenius metric. The RGD algorithm reads
\begin{equation}\label{RGD}
\bm{W}_{l+1}=\exp_{\bm{W}_l}\left(-\eta_l\operatorname{grad}\psi(\bm{W}_l)\right),
\end{equation}
where $\eta_l$ is the step size, and $\exp:T\operatorname{St}(M,\mathbb{C}^N)\rightarrow \operatorname{St}(M,\mathbb{C}^N)$ is the exponential map associated to the Euclidean metric of the Stiefel manifold. For more details, the reader may refer to \cite{AMS2008,HPL2021}.

To compile the RGD algorithm \eqref{RGD}, the Euclidean gradient of the function $\psi(\bm{W})$ is needed. Note that the Frobenius metric \eqref{eq:glmetric} can be extended to $N \times M$ matrices, namely
\begin{equation}\label{eq:fmst}
\langle \bm{X},\bm{Y}\rangle =\operatorname{tr}\left(\bm{X}^{\operatorname{H}}\bm{Y}\right),\quad \bm{X},\bm{Y}\in\mathbb{C}^{N\times M}.
\end{equation}

%In the following Propositions \ref{prop:LEMgra}, \ref{prop:JBLDgra} and \ref{prop:SKLDgra}, we introduce the notation:
%\begin{equation}
%\bm{V}=\bm{W}^{\operatorname{H}}\bm{R}_i\bm{W}.
%\end{equation}

\begin{prop}\label{prop:LEMgra}
The Euclidean gradient of the function $\psi(\bm{W})$ defined by \eqref{eq:psi} associated with the LEM is given by
\begin{equation}
\begin{aligned}
\nabla \psi(\bm{W}) = &-\frac{4}{J+K} \sum_{i=1}^{J+K} \bm{R}_i\bm{W}\Big(\bm{V}^{-1}\operatorname{Log}\bm{V} \\
&\quad-\int_0^1[(\bm{V}-\bm{I})s+\bm{I}]^{-1} \left(\operatorname{Log}\widehat{\bm{Z}}_t \right) \\
&\quad\quad \quad\quad\quad \quad   \times    [(\bm{V}-\bm{I})s+\bm{I}]^{-1}\operatorname{d}\!s \Big),
\end{aligned}
\end{equation}
where $\bm{V}=\bm{W}^{\operatorname{H}}\bm{R}_i\bm{W}$.
\end{prop}

\begin{proof} See Appendix \ref{app:LEMgra}.
\end{proof}

\begin{prop}\label{prop:AIRMgra}
The Euclidean gradient of the function $\psi(\bm{W})$ with respect to the AIRM is given by
\begin{equation}
\begin{aligned}
\nabla \psi(\bm{W})&=\frac{4}{J+K}\sum_{i=1}^{J+K}\bm{R}_i\bm{W}\left(\bm{W}^{\operatorname{H}}\bm{R}_i\bm{W}\right)^{-1}\\
&\quad\quad\quad\quad\quad \quad \times \operatorname{Log}\left(\widehat{\bm{Z}}_t\left(\bm{W}^{\operatorname{H}}\bm{R}_i\bm{W}\right)^{-1}\right).%\\
%\nabla \psi(\bm{W})=\frac{4}{m+n}\sum_{i=1}^{m+n}\left(\bm{W}^{\operatorname{H}}\bm{R}_i\bm{W}\right)^{-1}\operatorname{Log}\left(\widehat{\bm{Z}}_t\left(\bm{W}^{\operatorname{H}}\bm{R}_i\bm{W}\right)^{-1}\right)\bm{W}^{\operatorname{H}}\bm{R}_i.
\end{aligned}
\end{equation}
%where again $\bm{V}=\bm{W}^{\operatorname{H}}\bm{R}_i\bm{W}$.
\end{prop}

\begin{proof}
The corresponding function $\psi(\bm{W})$ is
\begin{equation}
\psi(\bm{W})=-\frac{1}{J+K}\sum_{i=1}^{J+K}\operatorname{tr}\left(\operatorname{Log}^2\left(\left(\bm{W}^{\operatorname{H}}\bm{R}_i\bm{W}\right)^{-1}\widehat{\bm{Z}}_t\right)\right).
\end{equation} Using definition of the Euclidean gradient and following a similar proof of Appendix \ref{app:LEMgra}, the result can be directly obtained. Similar computation details are omitted here.

\end{proof}

\begin{prop}\label{prop:JBLDgra}
The Euclidean gradient of $ \psi(\bm{W})$ associated with the JBLD is given by
\begin{equation}
\begin{aligned}
\nabla \psi(\bm{W}) &=-\frac{1}{J+K}\sum_{i=1}^{J+K} \bm{R}_i\bm{W}\Big(2\left(\bm{W}^{\operatorname{H}}\bm{R}_i\bm{W}+\widehat{\bm{Z}}_t\right)^{-1} \\
&\quad \quad\quad \quad \quad \quad \quad\quad  -\left(\bm{W}^{\operatorname{H}}\bm{R}_i\bm{W}\right)^{-1}\Big).
\end{aligned}
\end{equation}
%where again $\bm{V}=\bm{W}^{\operatorname{H}}\bm{R}_i\bm{W}$.
\end{prop}

\begin{proof}
See Appendix \ref{app:JBLDgra}.
\end{proof}

\begin{prop}\label{prop:SKLDgra}
The Euclidean gradient of $ \psi(\bm{W})$ associated with the SKLD is given by
\begin{equation}
\begin{aligned}
\nabla \psi(\bm{W}) &= -\frac{1}{J+K} \sum_{i=1}^{J+K} \bm{R}_i\bm{W}\Big( \bm{\widehat{Z}}_t^{-1}   \\
&\quad \quad - \left(\bm{W}^{\operatorname{H}}\bm{R}_i\bm{W}\right)^{-1}\bm{\widehat{Z}}_t\left(\bm{W}^{\operatorname{H}}\bm{R}_i\bm{W}\right)^{-1}\Big).\\
\end{aligned}
\end{equation}
\end{prop}

\begin{proof}
See Appendix \ref{app:SKLDgra}.
\end{proof}

\subsection{Complexity Analysis}
In this subsection, we will briefly show the  complexity of the calculation of LEM, AIRM, JBLD, SKLD means given by Propositions \ref{prop:LEMm}, \ref{prop:AIRMm}, \ref{prop:JBLDm}, \ref{prop:SKLDm} and the arithmetic mean, as well as the Euclidean gradients of $\psi(\bm{W})$ with respect to the LEM, AIRM, JBLD, SKLD, respectively. The latter are given by Propositions \ref{prop:LEMgra}, \ref{prop:AIRMgra}, \ref{prop:JBLDgra}, \ref{prop:SKLDgra}. For simplicity, we only keep the leading terms; for numerical iterations, we only provide the computational complexity for a single step.

The complexity figures assume that  $K$ number of $N\times N$ HPD matrices are given and the arithmetic with individual elements has complexity $O(1)$. The lower HPD manifold is $M$-dimensional. The following facts are used: $\bm{R}^{-1}  \sim O(N^3)$ and $\operatorname{Log}\bm{R} \sim O(N^4)$. Matrix exponential in all algorithms only deal with Hermitian matrices, and one way to calculate their exponentials
is through eigenvalue decomposition, whose complexity is $O(N^3)$, same as that of matrix inversion.  %In general, matrix logarithm shall cost most time.

    \begin{table}[ht]
  \caption{Computational complexity of the means}
\centering
\begin{tabular}{|c|c|}
        \hline
          Geometric measures  & Complexity  \\
        \hline
        Arithmetic mean &  $O(N^2(K-1))$ \\
        \hline
    LEM (Proposition \ref{prop:LEMm})&  $ O(N^4K) $ \\
        \hline
       AIRM (Proposition \ref{prop:AIRMm}, per iteration)  &  $O(N^4(K-1))$  \\
        \hline
        JBLD  (Proposition \ref{prop:JBLDm}, per iteration) & $O(N^3(K+1))$  \\
        \hline
        SKLD (Proposition \ref{prop:SKLDm}) &  $O(N^3(K+6))$ \\
        \hline
    \end{tabular}
    \label{tab:m}
    \end{table}

    It is clear from TABLE \ref{tab:m} that the arithmetic mean costs least  time, followed by the SKLD mean. Although both of them are Riemannian distances, computation of the LEM mean is much faster than the AIRM mean.

     \begin{table}[ht]
  \caption{Computational complexity of the Euclidean gradients}
\centering
\begin{tabular}{|c|c|}
        \hline
        Geometric measures & Complexity for each step (mod $J+K$)  \\
        \hline
    LEM (Proposition \ref{prop:LEMgra})&  $ O(2M^4) +O(N^2M)$ \\
        \hline
       AIRM (Proposition \ref{prop:AIRMgra})  &  $O(M^4)+O(2N^2M)$  \\
        \hline
        JBLD  (Proposition \ref{prop:JBLDgra}) & $O(2M^3)+O(2NM^2)+O(2N^2M)$  \\
        \hline
        SKLD (Proposition \ref{prop:SKLDgra}) &  $O(4M^3)+O(2NM^2)+O(2N^2M)$ \\
        \hline
    \end{tabular}
    \label{tab:gra}
    \end{table}

    From TABLE \ref{tab:gra}, we notice that computation of gradients of the divergences, i.e., the JBLD and the  SKLD, costs less time  compared with the Riemannian distances, i.e., the LEM and the AIRM.  Main reason is again the latter depend on matrix logarithm.

% {\bf  Computation of means.}
%
% {\bf The mini-max optimization problem \eqref{eq:opt_problem}.}

\section{Simulation Results}
\label{sec:sim}

In this section, we perform simulations to verify the performance advantage of the detectors proposed in the current paper, which are compared with the state-of-the-art counterparts. %In particular, we first introduce the environment setup of these simulations as well as the implementation details, and then, we present relevant simulation results with discussions.}
%\vspace{-0.2cm}
\subsection{Environment Setup}

 The simulations are performed in a non-homogeneous clutter, specifically, in a Gaussian clutter in the presence of interferences. We generate the sample data by resorting to an $N$-dimensional complex circular Gaussian distribution with zero mean and the known covariance matrix
\begin{equation}
\bm{C} = \sigma_c^2\bm{C}_0 + \sigma_n^2 \bm{I},
\end{equation}
where $\sigma_c^2\bm{C}_0$ denotes the clutter with $\sigma_c^2$ the clutter power while $\sigma_n^2\bm{I}$ is the thermal noise with $\sigma_n$ the noise power. Therefore, the clutter-to-noise ratio (CNR) is given by \begin{equation}
CNR=\frac{\sigma_c^2}{\sigma_n^2}.
\end{equation}
 The structure of the CCM $\bm{C}_0$ is Gaussian shaped with one-lag correlation coefficient $\rho$, whose entries are given by
\begin{equation}
[\bm{C}_0]_{i,j} = \rho^{| i-j|} \exp\left(\operatorname{i}2\pi f_c (i-j)\right), \quad i,j=1,2,\ldots,N.
\end{equation}
Here, $f_c$ is the normalized Doppler frequency.  $K$ secondary HPD matrices derived from the diagonal loading formalism \eqref{eq::diagonal_loading} are employed to estimate the CCM matrix as $\bm{R}_\mathcal{G}$. The HPD matrix $\bm{R}_D$ in the CUT is computed by the sample data $\bm{x}_D$. In the following, the parameters are chosen as $\sigma_n^2 = 1$, CNR$=25$ dB, $\rho = 0.95$ and $f_c = 0.1$.

\subsection{The Training Data}

%\begin{figure}[!t]
%\centering
%\subfigure[AIRM] {\includegraphics[width=8.5cm,angle=0]{Fig/Distance_AIRM}}
%\subfigure[LEM]  {\includegraphics[width=8.5cm,angle=0]{Fig/Distance_LEM}}
%\subfigure[JBLD] {\includegraphics[width=8.5cm,angle=0]{Fig/Distance_JBLD}}
%\subfigure[SKLD] {\includegraphics[width=8.5cm,angle=0]{Fig/Distance_SKLD}}
%\caption{Clutter-target distance for different measures}
%\label{Distance}
%\end{figure}

The dimension of the sample data is set to be $N=8$.  The normalized Doppler frequency of target signal is set to $f_s = 0.2$.  Two interferences are injected into the secondary data with the normalized Doppler frequency $f = 0.22$. The training dataset  consists of two subsets with the size of $2000$ each: the set of CCM and the set of HPD matrices containing a signal with SCR$=25$ dB. Fig. \ref{Distance} shows the distance between the CCMs and the HPD matrices with a target signal for different measures.  In particular,  within the cases of AIRM, LEM and  JBLD, the clutter-clutter distances are more scattered compared with the  clutter-target distances, while conversely the clutter-target distances are more scattered under the SKLD.

\begin{figure*}[htbp]
\centering
\includegraphics[width=16cm,angle=0]{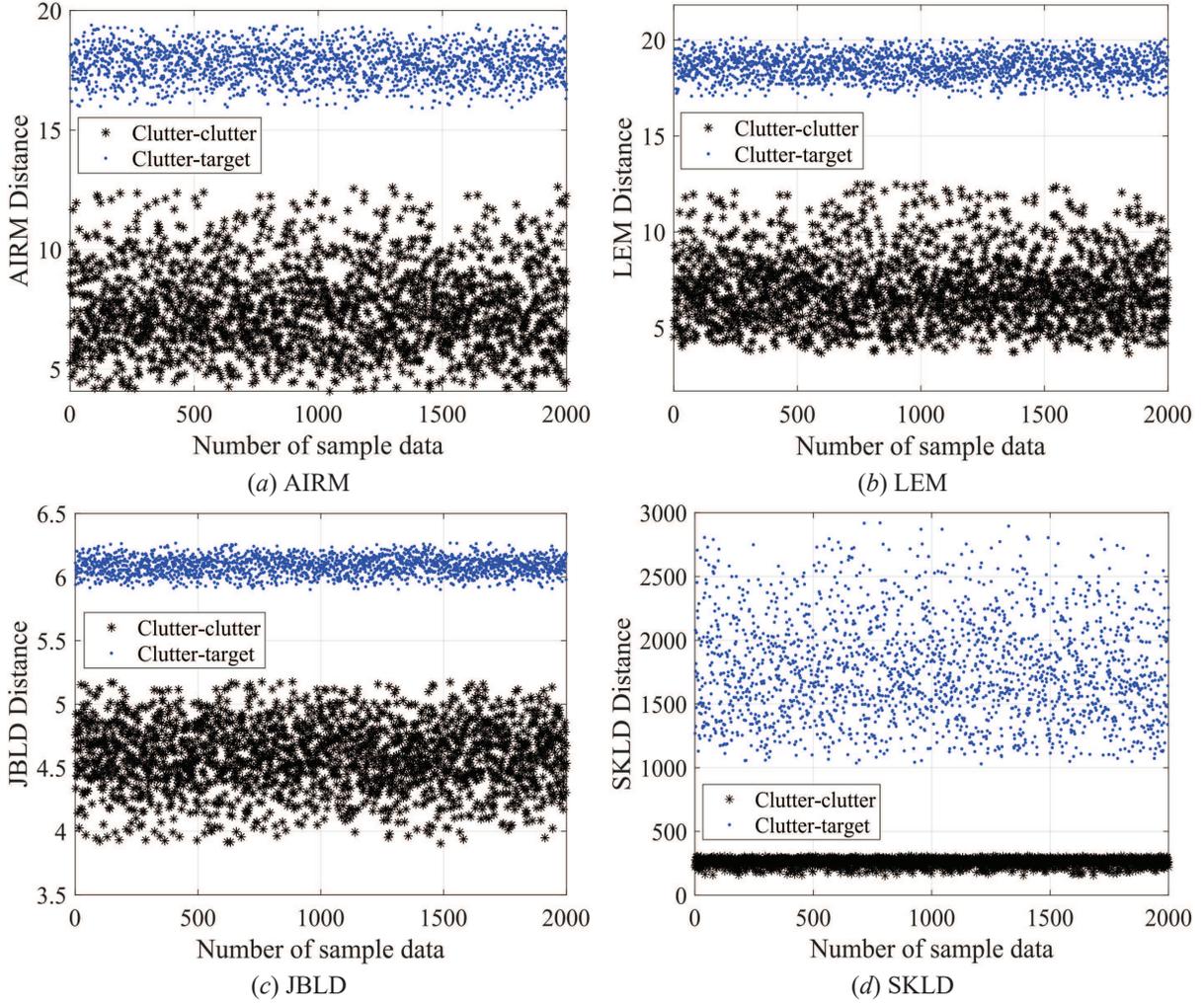}
\caption{Clutter-target distance for different measures}
\label{Distance}
%\vspace{-0.7cm}
\end{figure*}
%\vspace{-2cm}

\subsection{Comparison with Relevant Algorithms}

In order to verify the effectiveness of the proposed detectors, we compare the proposed methods with several well-received algorithms. For convenience, the following abbreviations are adopted.
\begin{itemize}
  \item AMF: The traditional adaptive matched filter \cite{4104190}.
  \item Benchmark: The AMF with known covariance matrix. It is the optimal performance for the types of AMFs.
  \item TBD-MIG detectors: The MIG detectors with the total square loss (TSL), the total von-Neumann (TVN) divergence, and the total log-determinant (TLD) divergence \cite{HUA2018232,huaetal2020}.
\end{itemize}

Unlike the AMF, the optimal performance of MIG detectors is not the MIG detectors with the known CCM since the detection performance is closely related to the discrimination between the target signal and the clutter. To decrease the computational load, we choose  the probability of false alarm as $P_{fa}=10^{-3}$. A number of $100/P_{fa}$ independent trials are repeated to estimate  the threshold, while  $2000$ independent trials are repeated to estimate  the probability of detection $P_d$.

\subsection{Simulation Results and Discussions}

By using the training dataset, we derive the three projection matrices that transforms the $N\times N$ HPD matrices to $M\times M$ HPD matrices for $M = 8$, $6$, $4$, and $2$, respectively.  We then perform the signal detection on these manifolds for different size of $K$ secondary data, where $K = M, 1.5M$, and $2M$, respectively. Statistically, as $K$ increases, the estimate accuracy of the CCM improves, that will certainly affect the detection performance. Figs. \ref{K_N}, \ref{K_15_N} and \ref{K_2_N} plot the $P_d$ vs SCR for the proposed MIG detectors and their corresponding counterparts as well as the TBD-MIG detectors and the AMF under different sizes of secondary data. The AMF with the known CCM is also provided as a benchmark. Figs. \ref{K_N}, \ref{K_15_N} and \ref{K_2_N} show that the detection performances of all the considered detectors improve as $K$ becomes larger. In Fig. \ref{K_N}, the MIG detectors can still work well while the $P_d$ of the AMF is very low, because that the estimate accuracy of the SCM is worse when $K=M$. It should  also be noticed that all the MIG detectors with manifold projection have better performances compared with their unprojected counterparts, namely the original MIG detectors, and both the projected and unprojected MIG detectors outperform the AMF except for the SKLD-MIG detector under $K=2M$.  In other words, the manifold projection can promote the discriminative power of HPD matrices. Moreover, the TBD-MIG detectors outperform the unprojected AIRM and LEM MIG detectors and both the projected and unprojected SKLD MIG detectors.

%\begin{figure*}[!t]
%\centering
%\subfigure[AIRM]  {\includegraphics[width=8.5cm,angle=0]{Fig/AIRM_Gaussian_8}}
%\subfigure[LEM]  {\includegraphics[width=8.5cm,angle=0]{Fig/LEM_Gaussian_8}}
%\subfigure[JBLD] {\includegraphics[width=8.5cm,angle=0]{Fig/JBLD_Gaussian_8}}
%\subfigure[SKLD] {\includegraphics[width=8.5cm,angle=0]{Fig/SKLD_Gaussian_8}}
%\caption{$P_d$ vs SCR, $K=M$}
%\label{K_N}
%\end{figure*}

\begin{figure*}[htbp]
\centering
\includegraphics[width=16cm,angle=0]{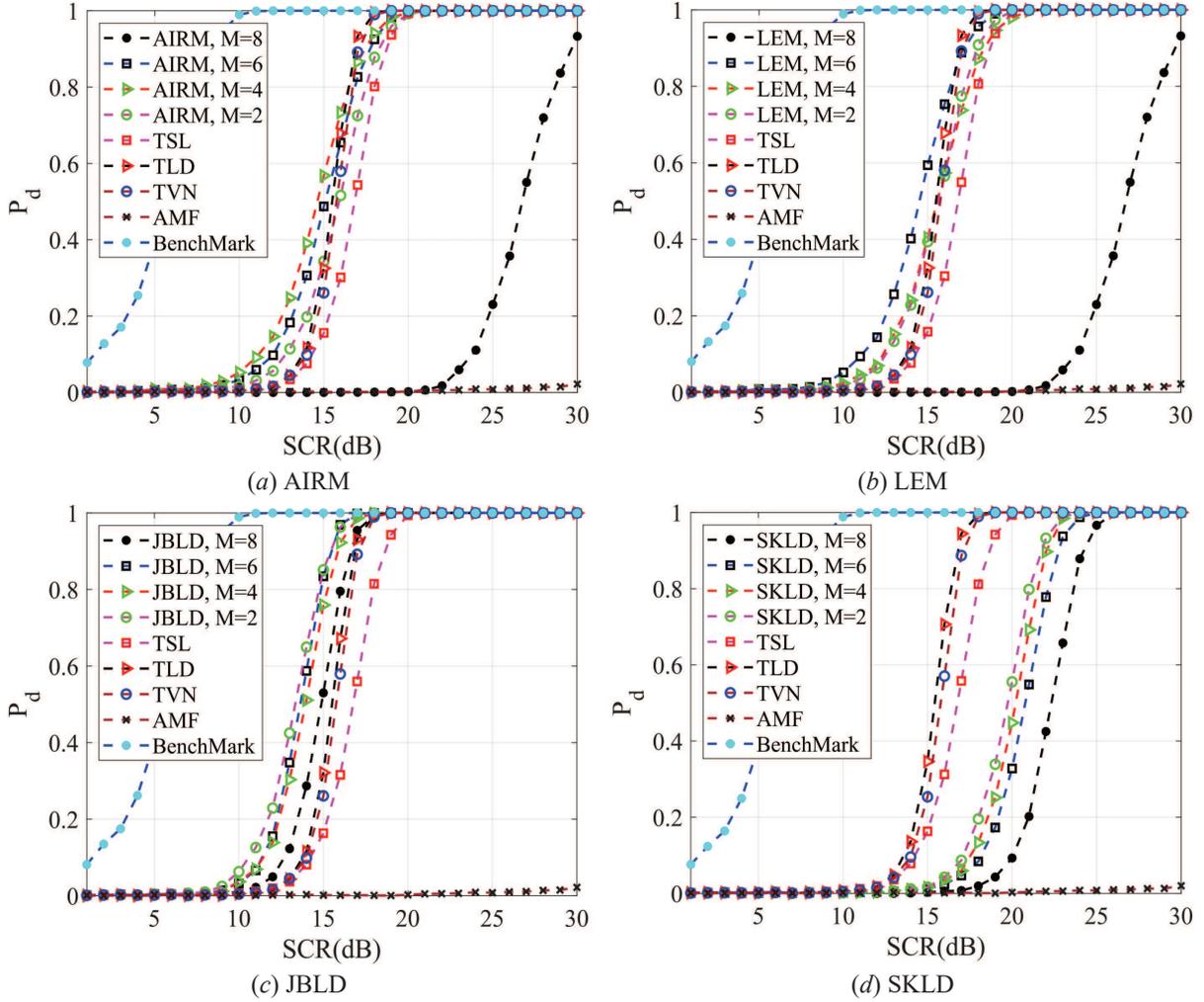}
\caption{The plots of $P_d$ vs SCR for $K=M$ in the nonhomogeneous clutter with two interferences.}
\label{K_N}
%\vspace{-0.7cm}
\end{figure*}

%\begin{figure*}[!t]
%\centering
%\subfigure[AIRM]  {\includegraphics[width=8.5cm,angle=0]{Fig/AIRM_Gaussian_12}}
%\subfigure[LEM]  {\includegraphics[width=8.5cm,angle=0]{Fig/LEM_Gaussian_12}}
%\subfigure[JBLD] {\includegraphics[width=8.5cm,angle=0]{Fig/JBLD_Gaussian_12}}
%\subfigure[SKLD] {\includegraphics[width=8.5cm,angle=0]{Fig/SKLD_Gaussian_12}}
%\caption{$P_d$ vs SCR, $K=1.5M$}
%\label{K_15_N}
%\end{figure*}

\begin{figure*}[htbp]
\centering
\includegraphics[width=16cm,angle=0]{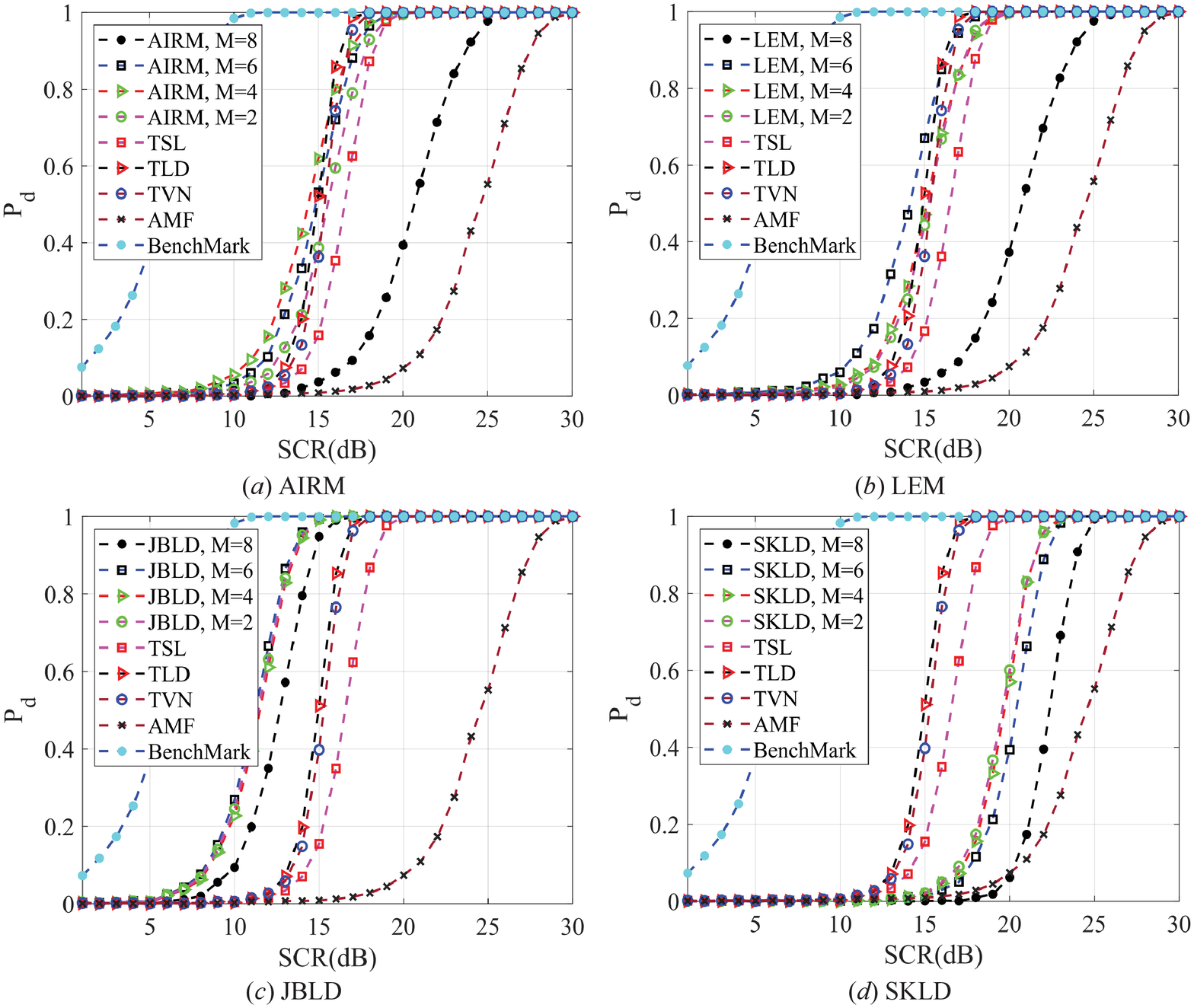}
\caption{The plots of $P_d$ vs SCR for $K=1.5M$ in the nonhomogeneous clutter with two interferences.}
\label{K_15_N}
%\vspace{-0.7cm}
\end{figure*}
%\vspace{-2.2cm}

%\begin{figure*}[!t]
%\centering
%\subfigure[AIRM]  {\includegraphics[width=8.5cm,angle=0]{Fig/AIRM_Gaussian_16}}
%\subfigure[LEM]  {\includegraphics[width=8.5cm,angle=0]{Fig/LEM_Gaussian_16}}
%\subfigure[JBLD] {\includegraphics[width=8.5cm,angle=0]{Fig/JBLD_Gaussian_16}}
%\subfigure[SKLD] {\includegraphics[width=8.5cm,angle=0]{Fig/SKLD_Gaussian_16}}
%\caption{$P_d$ vs SCR, $K=2M$}
%\label{K_2_N}
%\end{figure*}

\begin{figure*}[htbp]
\centering
\includegraphics[width=16cm,angle=0]{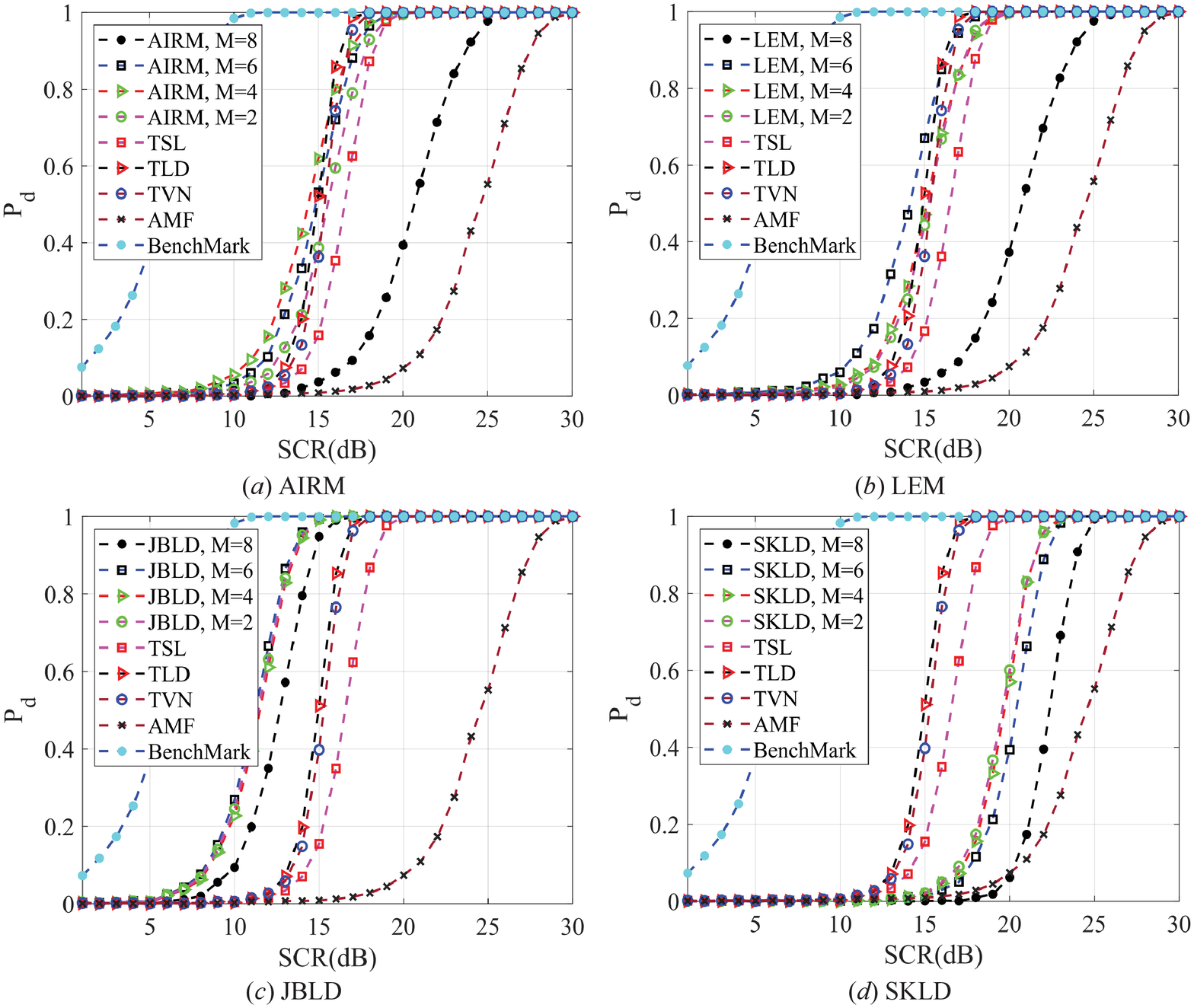}
\caption{The plots of $P_d$ vs SCR for $K=2M$ in the nonhomogeneous clutter with two interferences.}
\label{K_2_N}
%\vspace{-0.7cm}
\end{figure*}
%\vspace{-1.5cm}

To analyze the difference in the detection performance for different measure-based MIG detectors. Fig. \ref{Measures} shows the results of $P_d$ vs SCRs for different measures. It is obvious that the JBLD MIG detector has the best performance. Detection performance of the AIRM, LEM, and TBD is similar and they are better than the SKLD when $K>M$. It should be noted that it is probably difficult to determine detector which is universally better compared with the others since performance of the detection methods can also depend on features of the clutter. One important future research would be determining  the best detector against a specific type of clutter.

\begin{figure}[htbp]
\centering
\subfigure[$K=N$]  {\includegraphics[width=8.6cm,angle=0]{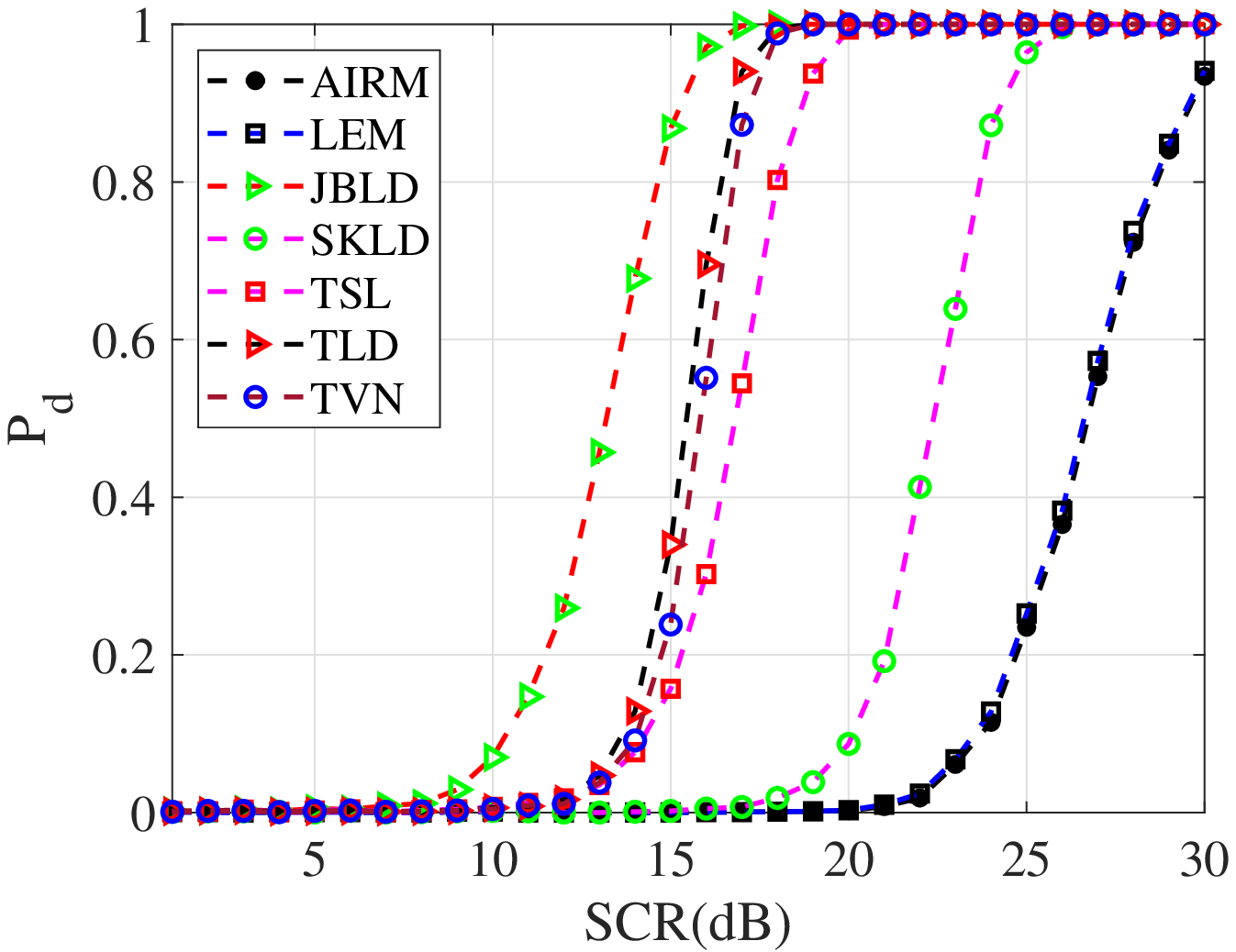}}
\subfigure[$K=1.5N$]  {\includegraphics[width=8.6cm,angle=0]{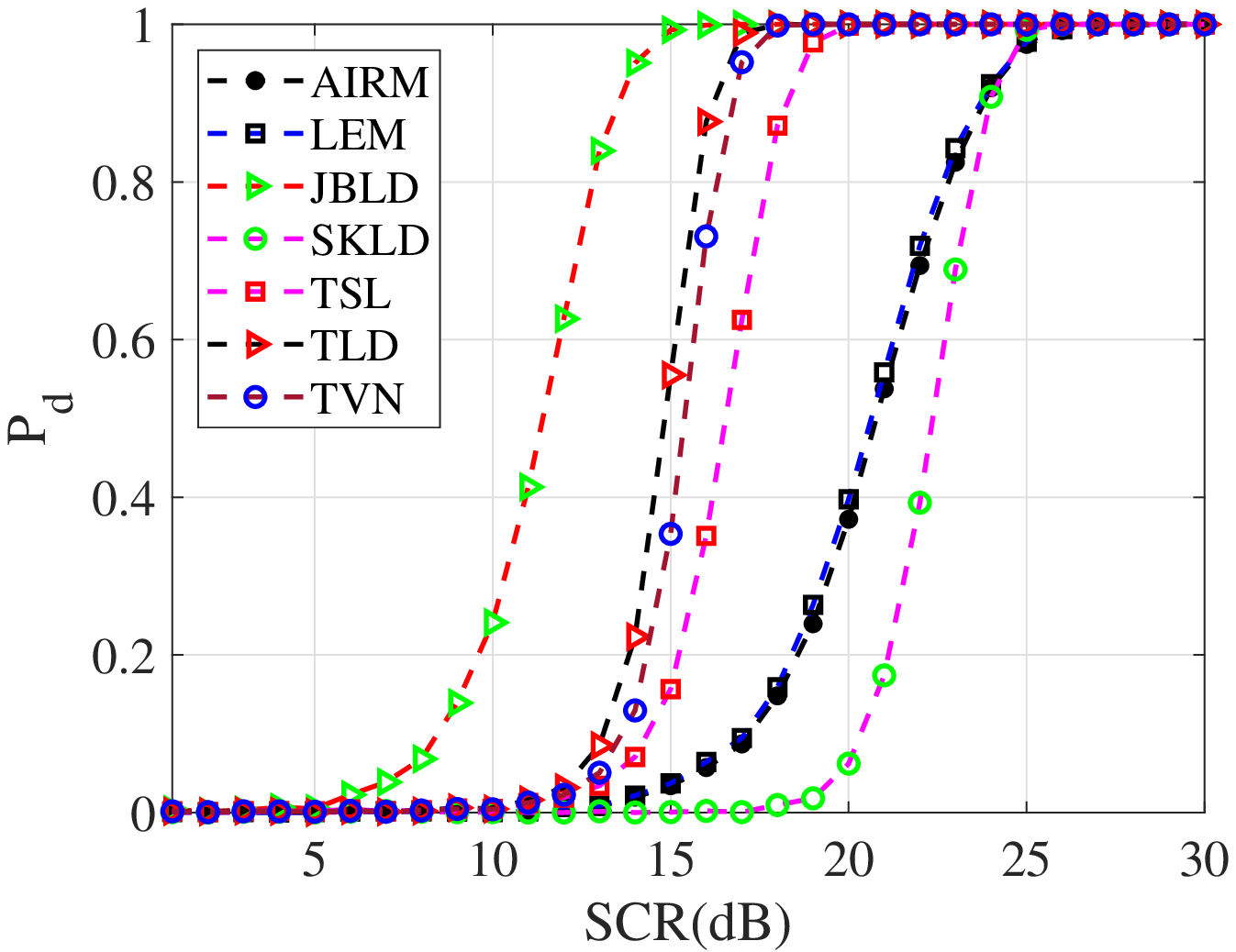}}
\subfigure[$K=2N$] {\includegraphics[width=8.6cm,angle=0]{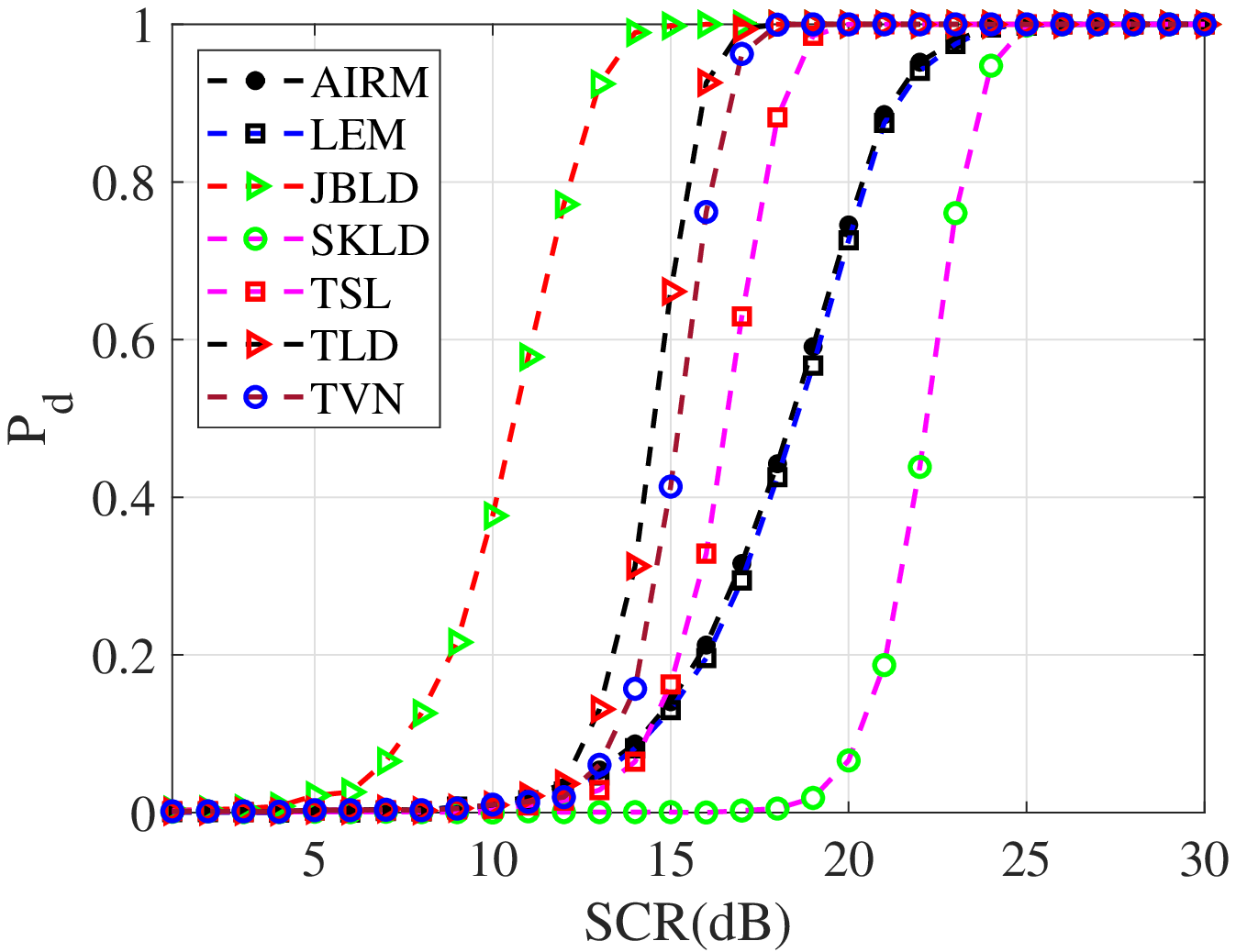}}
\caption{$P_d$ vs SCR for different measures}
\label{Measures}
%\vspace{-0.7cm}
\end{figure}

\section{Conclusions}
\label{sec:con}
In this paper, we proposed a class of learning discriminative MIG detectors in the unsupervised scenario, and applied them for signal detection in nonhomogeneous clutter. The sample data was interpreted as an HPD matrix, and the secondary HPD matrices were used to estimate the CCM. Inspired by the principle of PCA, we constructed a manifold projection that maps higher-order HPD matrices to a more discriminative lower-dimensional HPD manifold with maximum data variance. Learning the projection with maximum data variance   could be formulated as a two-step mini-max optimization problem in the Stiefel manifold and the lower-dimensional HPD manifold, respectively, which was solved by the RGD algorithm.  Four discriminative MIG detectors were designed in the lower-dimensional  manifold with respect to the LEM distance, the AIRM distance, the JBLD and SKLD, respectively. Simulation results  showed that the proposed MIG detectors could outperform their state-of-the-art counterparts and the AMF in nonhomogeneous clutter.

Potential future research includes the distributed target detection and further studies of optimization problems in Riemannian manifolds of matrices, such as the Stiefel manifold in the current study. This is certainly based on deep understanding of the geometric structures of these manifolds.  Practical applications to radar or sonar via the discriminative MIG detectors  should be interesting as well.

\appendices
\section{Proof of Proposition \ref{prop:SKLDm}: The SKLD mean}
\label{app:SKLD}

The SKLD mean of  HPD matrices $\{\bm{R}_k \}_{k \in [K]}$
% $\{\bm{R}_1, \bm{R}_2, \ldots, \bm{R}_K \}$
is the minimizer of the function
\begin{equation*}
\begin{aligned}
F(\bm{R})&=\sum_{k=1}^Kd_S^2(\bm{R}_k,\bm{R})\\
&=\sum_{k=1}^K\operatorname{tr}\left(\bm{R}_k^{-1}\bm{R}+\bm{R}^{-1}\bm{R}^{-1}\right)-2NK
\end{aligned}
\end{equation*}
defined in $\mathscr{P}(N,\mathbb{C})$. Using the definition \eqref{def:gra}, gradient of the function with respect to the Frobenius metric can be obtained as
\begin{equation*}
\nabla F(\bm{R})=\sum_{k=1}^K \left(\bm{R}_k^{-1}-\bm{R}^{-1}\bm{R}_k\bm{R}^{-1}\right).
\end{equation*}
The stationary condition $\nabla F(\widehat{\bm{R}})=0$, i.e.,
\begin{equation*}
\sum_{k=1}^K\bm{R}_k^{-1}=\widehat{\bm{R}}^{-1}\left(\sum_{k=1}^K\bm{R}_k\right)\widehat{\bm{R}}^{-1},
\end{equation*}
can be rewritten, by multiplying $\widehat{\bm{R}}$ from the left and the right simultaneously, as
\begin{equation*}
\widehat{\bm{R}}\left(\sum_{k=1}^K\bm{R}_k^{-1}\right)\widehat{\bm{R}}=\sum_{k=1}^K\bm{R}_k.
\end{equation*}
It is a special (continuous time) algebraic Riccati equation
\begin{equation}\label{eq:ARE}
\widehat{\bm{R}}\bm{A}\widehat{\bm{R}}=\bm{B},
\end{equation}
where $\widehat{\bm{R}}\in\mathscr{P}(N,\mathbb{C})$ is unknown and the known coefficient matrices $\bm{A}$ and $\bm{B}$ are both HPD:
\begin{equation}\label{eq:AB}
\bm{A}=\frac{1}{K}\sum_{k=1}^K\bm{R}_k^{-1},\quad \bm{B}=\frac{1}{K}\sum_{k=1}^K\bm{R}_k.
\end{equation}
Next we will solve the Eq. \eqref{eq:ARE}. Multiplying by $\bm{A}^{1/2}$ on both sides, we have
\begin{equation*}
\bm{A}^{1/2}\widehat{\bm{R}}\bm{A}\widehat{\bm{R}}\bm{A}^{1/2}=\bm{A}^{1/2}\bm{B}\bm{A}^{1/2}.
\end{equation*}
Noticing that its left-hand-side is exactly
\begin{equation*}
\left(\bm{A}^{1/2}\widehat{\bm{R}}\bm{A}^{1/2}\right)^2,
\end{equation*}
we obtain
\begin{equation*}
\bm{A}^{1/2}\widehat{\bm{R}}\bm{A}^{1/2}=\left(\bm{A}^{1/2}\bm{B}\bm{A}^{1/2}\right)^{1/2},
\end{equation*}
and hence
\begin{equation*}
\widehat{\bm{R}}=\bm{A}^{-1/2}\left( \bm{A}^{1/2}\bm{B}\bm{A}^{1/2} \right)^{1/2}\bm{A}^{-1/2},
\end{equation*}
where $\bm{A}$ and $\bm{B}$ are given by \eqref{eq:AB}.
This completes the proof.

\section{Proof of Proposition \ref{prop:LEMgra}}
\label{app:LEMgra}

At step $t$, the loss function with respect to the LEM reads
\begin{equation*}
\begin{aligned}
\psi(\bm{W}) &= -\frac{1}{J+K} \sum_{i=1}^{J+K} \norm{\operatorname{Log}\left(\bm{W}^{\operatorname{H}}\bm{R}_i\bm{W}\right)-\operatorname{Log}\bm{\widehat{Z}}_t}^2.
%&=- \frac{1}{m+n} \sum_{i=1}^{m+n} \operatorname{tr} \Big( \big(\operatorname{Log}(\bm{W}^{\operatorname{H}}\bm{R}_i\bm{W})-\operatorname{Log}(\bm{\widehat{R}}_t)\big)^2 \Big).
\end{aligned}
\end{equation*}
It suffices to show the Euclidean gradient of the function $F_i(\bm{W})$ with respect to the extended Frobenius metric \eqref{eq:fmst},
where
\begin{equation*}
F_i(\bm{W})=\norm{\operatorname{Log}\left(\bm{W}^{\operatorname{H}}\bm{R}_i\bm{W}\right)-\operatorname{Log}\bm{\widehat{Z}}_t}^2.
\end{equation*}

Writing
\begin{equation}\label{eq:A}
\bm{A}(\varepsilon)=(\bm{W}+\varepsilon \bm{X})^{\operatorname{H}}\bm{R}_i(\bm{W}+\varepsilon \bm{X})
\end{equation}
and using Eq. \eqref{def:gra}, for an ${N\times M}$ matrix $\bm{X}$, we have
\begin{equation*}
\begin{aligned}
\langle \nabla &F_i(\bm{W}),\bm{X}\rangle = \frac{\operatorname{d}}{\operatorname{d}\!\varepsilon}\Big|_{\varepsilon=0}F_i(\bm{W}+\varepsilon \bm{X})\\
& =\frac{\operatorname{d}}{\operatorname{d}\!\varepsilon}\Big|_{\varepsilon=0}\operatorname{tr}\left(\operatorname{Log}\bm{A}(\varepsilon)-\operatorname{Log}\widehat{\bm{Z}}_t\right)^2\\
&=2\operatorname{tr}\left(\left(\operatorname{Log}\bm{A}(0)-\operatorname{Log}\widehat{\bm{Z}}_t\right)\frac{\operatorname{d}}{\operatorname{d}\!\varepsilon}\Big|_{\varepsilon=0}\operatorname{Log}\bm{A}(\varepsilon)\right).
\end{aligned}
\end{equation*}
Noticing
\begin{equation*}
\frac{\operatorname{d}}{\operatorname{d}\!\varepsilon}\Big|_{\varepsilon=0}\bm{A}(\varepsilon)= \bm{X}^{\operatorname{H}}\bm{R}_i\bm{W}+\bm{W}^{\operatorname{H}}\bm{R}_i \bm{X}
\end{equation*}
and applying Lemma \ref{lem:aa}, the above equality becomes
\begin{equation*}
\begin{aligned}
\langle \nabla F_i&(\bm{W}),\bm{X}\rangle =4\operatorname{tr}\Big(\left(\operatorname{Log}\bm{A}(0)-\operatorname{Log}\widehat{\bm{Z}}_t\right)\\
&\times \int_0^1
[(\bm{A}(0)-\bm{I})s+\bm{I}]^{-1}\bm{W}^{\operatorname{H}}\bm{R}_i \bm{X}  \\
&\quad\quad\quad\quad\times [(\bm{A}(0)-\bm{I})s+\bm{I}]^{-1}\operatorname{d}\!s\Big)\\
=\: &4\operatorname{tr}\Big(\Big(\bm{A}^{-1}(0) \operatorname{Log}\bm{A}(0) -\int_0^1[(\bm{A}(0)-\bm{I})s+\bm{I}]^{-1} \\
&\times \left(\operatorname{Log}\widehat{\bm{Z}}_t \right)  [(\bm{A}(0)-\bm{I})s+\bm{I}]^{-1}\operatorname{d}\!s \Big) \bm{W}^{\operatorname{H}}\bm{R}_i \bm{X}\Big).
\end{aligned}
\end{equation*}
From the defintion $\langle \nabla F_i(\bm{W}),\bm{X}\rangle=\operatorname{tr}\left(\left(\nabla F_i(\bm{W})\right)^{\operatorname{H}}\bm{X}\right)$, we immediately have that
\begin{equation*}
\begin{aligned}
\nabla F_i(\bm{W})&=4\bm{R}_i\bm{W}\Big(\bm{A}^{-1}(0) \operatorname{Log}\bm{A}(0) \\
&\quad-\int_0^1[(\bm{A}(0)-\bm{I})s+\bm{I}]^{-1} \left(\operatorname{Log}\widehat{\bm{Z}}_t \right)\\
&\quad\quad\quad \quad \times   [(\bm{A}(0)-\bm{I})s+\bm{I}]^{-1}\operatorname{d}\!s \Big),
\end{aligned}
\end{equation*}
where $\bm{A}(0)=\bm{W}^{\operatorname{H}}\bm{R}_i\bm{W}$, i.e.,  the matrix $\bm{V}$ in Proposition \ref{prop:LEMgra}. This completes the proof.

%\section{Proof of Proposition \ref{prop:AIRMgra}}
%\label{app:AIRMgra}
%In the case of AIRM, the loss function can be written as
%\begin{equation*}
%\psi(\bm{W})=-\frac{1}{m+n}\sum_{i=1}^{m+n}F_i(\bm{W}),
%\end{equation*}
%where
%\begin{equation*}
%F_i(\bm{W})=\operatorname{tr}\left(\operatorname{Log}^2\left(\bm{A}^{-1}(0)\widehat{\bm{Z}}_t\right)\right).
%\end{equation*}
%Here $\bm{A}(0)$ is given by  \eqref{eq:A}. Using definition of the Euclidean gradient and following a similar proof of Appendix \ref{app:LEMgra}, the result can be directly obtained. We omit the similar computation details here.

\section{Proof of Proposition \ref{prop:JBLDgra}}
\label{app:JBLDgra}

The following lemma will be used.

\begin{lem}\label{lem:det}
For any invertible matrix $\bm{B}(\varepsilon)$, we have
\begin{equation*}
\frac{\operatorname{d}}{\operatorname{d}\!\varepsilon}\det \bm{B}(\varepsilon)=\det \bm{B}(\varepsilon) \operatorname{tr}\left(\bm{B}^{-1}(\varepsilon)\frac{\operatorname{d}}{\operatorname{d}\!\varepsilon}\bm{B}(\varepsilon)\right).
\end{equation*}
\end{lem}

Now the loss function can be written as
\begin{equation*}
\psi(\bm{W})=-\frac{1}{J+K}\sum_{i=1}^{J+K}F_i(\bm{W})
\end{equation*}
with
\begin{equation*}
F_i(\bm{W})=\ln\det\left(\frac{\bm{A}(0)+\widehat{\bm{Z}}_t}{2}\right)-\frac12\ln\det\left(\bm{A}(0)\widehat{\bm{Z}}_t\right),
\end{equation*}
where $\bm{A}(\varepsilon)$ is given by Eq. \eqref{eq:A} and  $\bm{A}(0)=\bm{W}^{\operatorname{H}}\bm{R}_i\bm{W}$. Using Lemma \ref{lem:det}, definition of the Euclidean gradient gives
\begin{equation*}
\begin{aligned}
\langle \nabla F_i&(\bm{W}),\bm{X}\rangle = \frac{\operatorname{d}}{\operatorname{d}\!\varepsilon}\Big|_{\varepsilon=0}\ln\det\left(\frac{\bm{A}(\varepsilon)+\widehat{\bm{Z}}_t}{2}\right)\\
&\quad\quad\quad  -\frac12\frac{\operatorname{d}}{\operatorname{d}\!\varepsilon}\Big|_{\varepsilon=0}\ln\det\left(\bm{A}(\varepsilon)\widehat{\bm{Z}}_t\right)\\
&=\operatorname{tr}\left(\left(2\left(\bm{A}(0)+\widehat{\bm{Z}}_t\right)^{-1}-\bm{A}^{-1}(0)\right)\bm{W}^{\operatorname{H}}\bm{R}_i\bm{X}\right).
\end{aligned}
\end{equation*}
Consequently, we have
\begin{equation*}
%\begin{aligned}
 \nabla F_i(\bm{W})=\bm{R}_i\bm{W}\left(2\left(\bm{A}(0)+\widehat{\bm{Z}}_t\right)^{-1}-\bm{A}^{-1}(0)\right)
 \end{equation*}
and 
\begin{equation*}
\nabla \psi(\bm{W})=-\frac{1}{J+K}\sum_{i=1}^{J+K}\nabla F_i(\bm{W}),
\end{equation*}
%\begin{equation*}
%\begin{aligned}
% \nabla F_i(\bm{W})=\bm{R}_i\bm{W}\left(2\left(\bm{A}(0)+\widehat{\bm{Z}}_t\right)^{-1}-\bm{A}^{-1}(0)\right),
%\end{aligned}
%\end{equation*}
%and
%\begin{equation*}
%\nabla \psi(\bm{W})=-\frac{1}{J+K}\sum_{i=1}^{J+K}\nabla F_i(\bm{W}),
%\end{equation*}
that finishes the proof.

\section{Proof of Proposition \ref{prop:SKLDgra}}
\label{app:SKLDgra}

Similarly, we write the loss function as
\begin{equation*}
\psi(\bm{W})=-\frac{1}{J+K}\sum_{i=1}^{J+K}F_i(\bm{W}),
\end{equation*}
where
\begin{equation*}
\begin{aligned}
F_i(\bm{W})& =\frac12 \operatorname{tr}\Bigg( \bigg( \bm{W}^{\operatorname{H}}\bm{R}_i\bm{W} \bigg)^{-1}\bm{\widehat{Z}}_t \\
&\quad \quad \quad \quad + \bm{\widehat{Z}}_t^{-1}\bigg( \bm{W}^{\operatorname{H}}\bm{R}_i\bm{W} \bigg) - 2\bm{I}\Bigg).
\end{aligned}
\end{equation*}
Euclidean gradient of the function $F_i(\bm{W})$ is given by
\begin{equation*}
\begin{aligned}
\langle \nabla F_i&(\bm{W}),\bm{X}\rangle = \frac12\frac{\operatorname{d}}{\operatorname{d}\!\varepsilon}\Big|_{\varepsilon=0}\operatorname{tr}\left(\bm{A}^{-1}(\varepsilon)\widehat{\bm{Z}}_t+\widehat{\bm{Z}}_t^{-1}\bm{A}(\varepsilon)\right)\\
&=\frac12 \operatorname{tr}\Big(\left(\widehat{\bm{Z}}^{-1}_t-\bm{A}^{-1}(0)\widehat{\bm{Z}}_t^{-1}\bm{A}^{-1}(0)\right)\\
&\quad\quad\quad\quad \times \left(\bm{X}^{\operatorname{H}}\bm{R}_i\bm{W}+\bm{W}^{\operatorname{H}}\bm{R}_i \bm{X}\right)\Big)\\
&=\operatorname{tr}\left(\left(\widehat{\bm{Z}}^{-1}_t-\bm{A}^{-1}(0)\widehat{\bm{Z}}_t\bm{A}^{-1}(0)\right)\bm{W}^{\operatorname{H}}\bm{R}_i \bm{X}\right),
\end{aligned}
\end{equation*}
and consequently, we obtain
\begin{equation*}
\nabla F_i(\bm{W})=\bm{R}_i\bm{W}\left(\widehat{\bm{Z}}^{-1}_t-\bm{A}^{-1}(0)\widehat{\bm{Z}}_t\bm{A}^{-1}(0)\right).
\end{equation*}
Here, $\bm{A}(\varepsilon)$ is given by \eqref{eq:A} and $\bm{A}(0)=\bm{W}^{\operatorname{H}}\bm{R}_i\bm{W}$. This completes the proof.

\ifCLASSOPTIONcaptionsoff
  \newpage
\fi

\bibliographystyle{IEEEtran}

\bibliography{mybibfile}

\end{document}